\documentclass[journal]{IEEEtran}
\title{MeshRef_ETC_Journal}
\IEEEoverridecommandlockouts
\usepackage{amsmath}
\usepackage{mathtools}
\usepackage{amsmath}
\usepackage{amssymb} 
\usepackage{mathtools}
\usepackage{amsthm}
\usepackage{algorithm}
\usepackage{algpseudocode}
\usepackage{subfig}
\usepackage{graphicx}
\usepackage{lipsum}
\usepackage{float}

\usepackage[table]{xcolor}
\usepackage[utf8]{inputenc}
\usepackage{fourier}
\usepackage{array}
\usepackage{makecell}
\usepackage{xr}
\usepackage{cite}
\usepackage{balance}

\usepackage{pgfplots}
\pgfplotsset{compat=1.12}
\usepgfplotslibrary{fillbetween}
\usetikzlibrary{patterns}

\mathtoolsset{showonlyrefs}

\renewcommand{\figurename}{Figure}

\newcommand{\norm}[1]{\left|#1\right|}
\newcommand{\vecnorm}[1]{\left\lVert#1\right\rVert}
\DeclarePairedDelimiter\ceil{\lceil}{\rceil}

\newtheorem{theorem}{Theorem}
\newtheorem{corollary}{Corollary}
\newtheorem{lemma}{Lemma}
\newtheorem{assumption}{Assumption}
\newtheorem{definition}{Definition}

\newtheorem{problem}{Problem}

\begin{document}

\title{Inaccuracy matters: accounting for solution accuracy in event-triggered nonlinear model predictive control}


\author{Omar J. Faqir, Eric C. Kerrigan
\thanks{This work has been submitted to the IEEE for possible publication. Copyright may be transferred without notice, after which this version may no longer be accessible.}
\thanks{The support of the EPSRC Centre for Doctoral Training in High Performance Embedded and Distributed Systems  (HiPEDS, Grant Reference EP/L016796/1) is gratefully acknowledged.
}
\thanks{O. J. Faqir is with the Department of Electrical \& Electronics Engineering, Imperial College London, SW7~2AZ, U.K. {\tt\small ojf12@ic.ac.uk}}%
\thanks{Eric C. Kerrigan is with the Department of Electrical \& Electronic Engineering
and Department of Aeronautics, Imperial College London, London SW7~2AZ, U.K. {\tt\small e.kerrigan@imperial.ac.uk}}%
}


\maketitle

\begin{abstract}
We consider the effect of using approximate system predictions in
event-triggered control schemes. Such approximations may result from
using numerical transcription methods for solving continuous-time
optimal control problems. Mesh refinement can guarantee upper bounds on the error in the differential equations which model the system dynamics. With the accuracy guarantees of a mesh refinement scheme, we show that the proposed event-triggering scheme --- which compares the measured system with approximate state predictions --- can be used with a guaranteed strictly positive inter-update time. We show that if we have knowledge of the employed transcription scheme or the approximation errors, then we can obtain better online estimates of inter-update times. We additionally detail a method of tightening constraints on the approximate system trajectory used in the nonlinear programming problem to guarantee constraint satisfaction of the continuous-time system. This is the first work to incorporate prediction accuracy in triggering metrics. Using the solution accuracy we can guarantee reliable lower bounds for inter-update times and perform solution dependent constraint tightening.
\end{abstract}


\IEEEpeerreviewmaketitle

\section{Introduction}

Periodic control policies are often accepted without question. 
A system measurement is taken at each sample period and used to calculate an appropriate control signal to apply to the plant. 
Although the choice of periodic control policies is in part due to the discrete nature of digital platforms, processor clock speed is typically several orders of magnitude faster than the dominant plant dynamics. 
More importantly, assumptions on sampling periodicity simplifies analysis \cite{aastrom1999comparison}.
Even so, the choice of sample period is non-obvious. It is often determined through engineering know-how (e.g. an order of magnitude faster than dominant dynamics) or based on the Nyquist-Shannon theorem. 
However, there is no explicit relationship between the Nyquist frequency and discretization accuracy, so sampling decisions are often conservative \cite{lazutkin2018approach}.

Periodic sampling policies are inherently open-loop.
Control practitioners determine a \emph{single} sample interval to satisfy performance and stability criteria for \emph{all} possible initial conditions and disturbance realizations \cite{gommans2014self}.
This `worst-case' periodic design leads to over allocation/utilization of communication and computation resources with no necessary improvement in control performance \cite{aastrom1999comparison}.
Furthermore, the systems of interest have also evolved --- we often don't have compute units dedicated for control. 
Embedded devices typically have a set of (real-time) tasks, of which control is only one. 
The problem of resource allocation is exacerbated for (wireless) networked control systems ((W)NCS) where sensors, actuators and controllers may be spatially separated \cite{henriksson2012self,mccann2004improved} and communication resources scarce.

The control loop is closed by measuring the plant, computing a control action, and updating the actuators.
Consider the following sets of times for each of these actions. 
(i) The evenly sampled grid $\mathcal{T}_g := \left\{t_k^g\right\}_{k\in\mathbb{N}}$, where $t_{k+1}^g-t_k^g=h, \forall k \in \mathbb{N}$, for some $h \geq 0$. By setting $h=0$ we can consider the archetype where the system is continuously monitored.
(ii) The set of times $\mathcal{T}_s := \left\{t_k^s\right\}_{k\in\mathbb{N}}$ where the system is sampled, and (iii) $\mathcal{T}_u = \left\{t_k^u\right\}_{k\in\mathbb{N}}$ where the control action is updated.
Without considering delays, periodic control system design implicitly assumes that $\mathcal{T}_g = \mathcal{T}_s = \mathcal{T}_u$, where $h>0$ is chosen sufficiently small to guarantee stability.
We feel that requiring the sets of update and sample times to be equivalent is overly restrictive. 
Intuitively, it is only required that $\mathcal{T}_u \subseteq \mathcal{T}_s \subseteq \mathcal{T}_g$.
The left inclusion follows since any feedback is parametric in the most recent state measurement.
If $\mathcal{T}_u \not \subseteq \mathcal{T}_s$ then the control is updated using outdated system information.
To characterize data losses and delays in NCS, \cite{peng2016designing} further distinguishes between times when the system measurements are taken, sent and received. 

The work \cite{gommans2014self} makes the case that any metric of control performance should reflect the offline and online cost of implementing the controller. 
The online cost is of particular importance in computationally demanding control methods, such as model predictive control (MPC).
In the exemplary work \cite{velasco2003self} a task model is proposed for joint determination of control and sampling of real-time systems. 
The principle is to increase/decrease control resources as a function of system state and the number of other processes requesting compute resources. 
The use of state feedback to dynamically determine sample times has since been central in the fields of event- and self-triggered control. 
These strategies, which are formulated in Section~\ref{sec:preliminaries}, only update the control when state-dependent conditions are met. 

(Nonlinear) model predictive control ((N)MPC) is an extensively used control technique, where a finite time control trajectory is designed to optimize some performance metric of the predicted system evolution. 
NMPC complements event- and self-triggering strategies, because the predictive information resulting from online optimization may allow for longer open-loop run times than using zero-order hold state feedback controllers \cite{lucia2016predictive}.
The computational burden of solving a nonlinear program (NLP) online is a key drawback. 
However, aperiodic triggering schemes have the potential to increase the average compute time available between control updates.

An additional consideration is that numerical methods for solving the continuous-time optimal control problems that arise in NMPC only generate approximations of the continuous state and input trajectories. 
The computational burden of NMPC may potentially be reduced by either solving a less complex NLP --- resulting in \emph{lower quality} predictions --- or by solving the NLP less frequently --- resulting in a more \emph{outdated} prediction.
The quality of the NMPC solution is different from closed-loop performance, and may by guaranteed through mesh refinement schemes. This paper is the first work to consider the effect of these approximation errors in triggering schemes.

We investigate how the solution accuracy and triggering conditions will affect the inter-update times in an event-triggered NMPC framework.
We bound the inter-update time (IUT) 
\begin{equation}
\min_{k\in\mathbb{N}} t_{k+1}^u - t_k^u    
\end{equation} 
using certifications on the solution accuracy, guaranteed through a mesh refinement scheme, and the triggering protocol. 
We further show how error information provided by online mesh refinement schemes can be used to generate less conservative IUT estimates. 
A method for constraint tightening is proposed to guarantee that constraints placed on the approximate system predictions at a finite number of time instances ensure the real system trajectory satisfies constraints.

In this work we will first introduce a set of preliminaries (Section~\ref{sec:preliminaries}) detailing basic event- and self- triggered control paradigms, the general nonlinear model predictive control optimization problem and a possible procedure for solving this problem, using direct collocation and mesh refinement. In Section~\ref{sec:MinIUT_Offline} we use the guarantees from event-triggering and mesh refinement schemes to determine a guaranteed minimum time between control updates, which depends on the allowable prediction error and solution accuracy. This bound is improved in Section~\ref{sec:OnlineTrig} where we use the actual error in the solution of the optimization problem, rendering less conservative online estimates of the inter control update times. In Sections~\ref{sec:LinearExample} and~\ref{sec:NonlinearExample} these methods are exemplified for linear and nonlinear systems. This work significantly extends results from \cite{faqir2020Mesh} by generalizing the system under consideration, the triggering conditions being employed and proposing constraint tightening schemes.

\section{Preliminaries} \label{sec:preliminaries}
Consider the nonlinear dynamical system described by
\begin{align} \label{eq:sysDyn}
\dot{x}(t) &= f(t,x(t),u(t)) + v(t),\\
y(t) &= x(t) + \theta(t)
\end{align}
where $x(t)\in \mathbb{R}^n$ denotes the state of the system, $y(t) \in \mathbb{R}^n$ the measured/estimated state, $u(t)\in \mathbb{R}^m$ the control input, and $v(t), \theta(t) \in \mathbb{R}^n$ are bounded exogenous inputs that satisfy $\vecnorm{v(t)} \leq \hat{v}, \vecnorm{\theta(t)} \leq \hat{\theta}$. The nominal system dynamics are described by $f(\cdot)$ and are Lipschitz in $x$ with constant~$L_x$, i.e.\
\begin{equation}
	\vecnorm{f(t,x_1,u) - f(t,x_2,u)} \leq L_x \vecnorm{x_1 - x_2}
\end{equation}
for all $x_1$, $x_2$, $t$, and $u$. Furthermore, $\vecnorm{f(\cdot,x,u)} \leq L_f$ for all $x,u$. 

\subsection{Event- and Self-Triggered Control}

In event-triggered control (ETC) (e.g. \cite{heemels2012introduction}), the system is continuously or periodically sampled, but the control is only updated when certain `triggering' conditions are met. Consider, for example, the following problem adapted from~\cite{heemels2012introduction}:

\begin{problem}[\textbf{ETC}] \label{prob:ETC}
	Consider the dynamical system described by \eqref{eq:sysDyn}, and a time-varying state feedback law $u(t)=\mu(x(t),t)$ which renders the closed-loop system globally asymptotically stable (GAS). Identify a set of state- and input- dependent conditions $F_\text{event} : \mathbb{R}^n \rightarrow \mathbb{R}$, resulting in update times
	\begin{equation} \label{eq:ETC_time}
	t_{i+1}^u = \inf \{t \in \mathbb{R}^+_0 | t > t_i^u, F_\text{event}(x(t)) \geq 0 \},
	\end{equation}
	for which the closed-loop system with sampled-data implementation is GAS (and satisfies appropriately defined performance criteria).
\end{problem}

Triggering based on predicted system evolution was used in \cite{heemels2013model} for stabilizing a linear networked control system (NCS), but assumed synchronization of predictions between controller and plant/sensor.
This was extended in \cite{wu2016redesigned} by introducing an artificial waiting mechanism, which guarantees stability in the presence of constant network delays. In NCS, predictive information allows for active compensation of delay and packet dropouts. The works \cite{yin2015model},\cite{yin2016model} use these compensation methods to combat imperfections in future transmitted data, giving stability conditions with and without plant-model mismatch.

Since $F_\text{event}(\cdot)$ depends on the current system state, ETC implementations often require dedicated sensor-side hardware to continuously monitor the state of the system.
Choice of an appropriate $F_\text{event}(\cdot)$ is still open and depends on the control scheme under consideration. We require the IUTs
    \begin{equation}
        \tau_i^u := t_{i+1}^u - t_i^u, \ \forall i \in \mathbb{R}^+
    \end{equation}
to be strictly positive. To avoid Zeno phenomena a lower bound $\underline{\tau}^u \leq \tau_i^u,\ \forall i \in \mathbb{R}^+$ must exist \cite{borgers2014event}.
The IUTs must also be upper bounded. MPC schemes often restrict $\tau_i^u$ to be less than the prediction horizon.
In addition to any stability/performance guarantees placed by the control practitioner, in the case of MPC we must guarantee constraint satisfaction and recursive feasibility of the optimal control problem (OCP).

In Problem~\ref{prob:ETC} the triggering condition is designed subject to an existing control law. This is sometimes referred to as emulation based design. However, in event-triggered MPC the problem formulation is often adjusted to guarantee stability and/or robustness  of the open-loop system for extended periods of time. As such, the problem and triggering formulations can often not be designed separately.
Some event-based controllers trigger directly on a measure of the difference between predicted and measured state. We explicitly denote triggering dependency on the most recent state predictions $\tilde{x}(\cdot)$ by defining and employing the triggering condition
\begin{equation}
\tilde{F}_\text{event}\left(x(t),\tilde{x}(t,t_k^u)\right).
\end{equation}

In self-triggered control (STC) (e.g. \cite{velasco2003self,heemels2012introduction}), the control is calculated at time $t_k^u$, from which the ystem evolution is predicted and an explicit IUT $\tau_k^u$ is calculated. Consider the following, adapted from \cite{anta2010sample}:

\begin{problem}[\textbf{STC}] \label{prob:STC}
	Consider the dynamical system described by \eqref{eq:sysDyn}, and time-varying state feedback law $u(t)=\mu(x(t),t)$ which renders the closed-loop system GAS. Identify a set of state and input dependent conditions $F_\text{self}:\mathbb{R}^n \times \mathbb{R}^+ \rightarrow \mathbb{R}$, resulting in inter-sample periods
	\begin{equation} \label{eq:STC_Prob_Condition}
	\tau_i^u = \inf \{\tau \in \mathbb{R}^+_0 | F_\text{self}(x(t_i^u),\tau) \geq 0 \},
	\end{equation}
	for which the closed-loop system is GAS (and satisfies appropriately-defined performance criteria).
\end{problem}

Both $F_\text{event}(\cdot)$ or $F_\text{self}(\cdot)$ are commonly chosen to guarantee some sufficient decrease in an appropriately defined Lyapunov function. 
Although STC lacks the inherent robustness of ETC, it does not require dedicated hardware for continuous sensing. 
In addition, $t_{i+1}^u$ is known $\tau_i^u$ seconds in advance in STC schemes.
This is important for process scheduling on embedded platforms or dynamic communication resource allocation. 
Both STC and ETC conditions are based on emulation of an \emph{a priori} known control law $\mu(\cdot)$. 
However, in the case of MPC, this emulation is implicit in the solution of the optimization problem \eqref{eq:Bolza_cont} defined below.
We will assume the prediction was always computed at the most recent update time, and to simplify notation will write $\delta(\cdot,t_k^u)$ as $\delta(\cdot)$ where possible.

\subsection{Nonlinear Model Predictive Control}
In MPC a finite time sequence/trajectory of control inputs is designed to optimize some performance metric of the predicted plant evolution \cite{mayne2000constrained}.
The optimization problem is parametric in the most recent state measurements. 
Typical MPC paradigms account for uncertainty by measuring/estimating the state and solving the associated optimization problem at periodic intervals.
Only the first portion of the computed actuation is applied to the system before the  procedure is repeated.
Application of MPC has historically been restricted to slow or simple (linear) dynamical systems due to the complexity of performing optimization close to real-time. However, greater computational resources and tailor-made optimization algorithms nowadays allow for MPC  of faster, more complex systems than in the past.

The continuous-time NMPC problem solved at time $t_i^u$ may be cast in the general Bolza form,
\begin{subequations}
	\makeatletter
	\def\@currentlabel{P}
	\makeatother
	\label{eq:Bolza_cont}
	\begin{align}
	\min_{\hat{x},\hat{u}} \ \Phi&(\hat{x}(t_0),\hat{x}(t_f)) + \int\limits_{t_0}^{t_f}L(\hat{x}(t),\hat{u}(t),t) \ \mathrm{d}t\tag{Pa} \label{eq:sc_cost} \\
	\text{s.t. } & \forall t \in [t_0,t_f], \notag\\
	& \dot{\hat{x}}(t) = f(\hat{x}(t),\hat{u}(t),t)  \text{ a.e.} \tag{Pb} \label{eq:dynODE} \\
	& c(\hat{x}(t),\hat{u}(t),t) \leq 0  \text{ a.e.} \tag{Pc} \label{eq:pathConstriants}\\
	& \phi(\hat{x}(t_0),\hat{x}(t_f)) = 0  \tag{Pd} \label{eq:termConstriants}\\
	&  \hat{x}(t_0)=x(t_i^u). \tag{Pe} \label{eq:IC_Meas}
	\end{align}
\end{subequations}
Optimization is performed over a horizon $T_\text{op}:=t_f-t_0$ over the internal variables~$\hat{x}, \hat{u}$, representing the state and input trajectory.
The cost \eqref{eq:sc_cost} is composed of the stage (Lagrange) cost functional $L(\cdot)$ and boundary (Mayer) cost functional~$\Phi(\cdot)$, which is evaluated at initial and final times $t_0,t_f$.
The resulting state trajectory, $x \in \mathcal{C}^0$ satisfies the nominal dynamics~$f(\cdot)$, typically enforced as ordinary differential equations~\eqref{eq:dynODE}. States and controls must satisfy the path constraints~\eqref{eq:pathConstriants} for all time.
Boundary conditions \eqref{eq:termConstriants} and $\Phi(\cdot)$ are common tools for guaranteeing closed-loop stability \cite{mayne2000constrained}.
The feedback occurs by including state measurements \eqref{eq:IC_Meas} as and initial condition.
Key advantages of MPC are its use of a predictive model and guarantees of constraint satisfaction.
We refer to $x^*(\cdot),u^*(\cdot)$ as the \emph{true} solution of the optimal control problem \eqref{eq:Bolza_cont} as
\begin{equation}
(x^*(\cdot),u^*(\cdot)) := \text{arg}\min_{(\hat{x},\hat{u})
} \eqref{eq:Bolza_cont}.
\end{equation}

\subsection{Direct Collocation \& Mesh Refinement}
The dynamic optimization problem \eqref{eq:Bolza_cont} is formulated in continuous-time, and is infinite dimensional.
Except for the simplest cases, \eqref{eq:Bolza_cont} is intractable or impossible to solve analytically. 
Direct collocation methods transcribe \eqref{eq:Bolza_cont} into a finite dimensional NLP, which is then solved using numerical optimization methods. The discrete NLP solution is interpolated to reconstruct an \emph{approximate} continuous-time solution $\tilde{x}(\cdot),u(\cdot)$ to \eqref{eq:Bolza_cont} \cite{betts2010practical}. 
System state trajectories may be approximated as continuous piece-wise polynomials  using a combination of $h$- and $p$-methods (known as $hp$-methods) \cite{kelly2017introduction}, briefly detailed below. In these methods the input is approximated as discontinuous piece-wise polynomials.
Define the interval $\Omega:=(t_0,t_f)$, and $\overline{\Omega}$ as the closure of $\Omega$.
\begin{definition} [Mesh]\label{def:Mesh}
	The set $\mathcal{T}_h$ is called a \textbf{mesh} and consists of open intervals $T\subset\Omega$ satisfying conditions
	\begin{enumerate}
		\item Disjunction $T_1 \cap T_2 = \emptyset \ \forall \text{ distinct } T_1,T_2 \in \mathcal{T}_h$,
		\item Coverage $\cup_{T\in\mathcal{T}_h} \overline{T} = \overline{\Omega}$,
		\item Resolution $\max_ {T\in\mathcal{T}_h}\norm{T} = h$,
		\item Quasi-uniformity $\min_{T_1,T_2\in\mathcal{T}_h} \frac{\norm{T_1}}{\norm{T_2}} \geq \sigma > 0$,
	\end{enumerate}
	where constant $\sigma$ must not depend on mesh parameter $h$.
\end{definition}


We denote the number of mesh segments by $K:=\operatorname{card}{\mathcal{T}_h}$ and refer to the $k^\text{th}$ segment as $T_k,\forall k \in \mathcal{K}_h := \{1,\hdots, K\}$.
We define the indexed set of mesh points as
\begin{equation}
\mathcal{T}_m := \cup_{k\in\mathcal{K}_h} \inf \overline{T}_k,
\end{equation}
where two polynomial segments are joined, and index them identically to the mesh segments.
The approximate state trajectory $\tilde{x}(\cdot)$ defined on this mesh will then satisfy
\begin{equation}
\tilde{x} \in \mathcal{X}_p:=\{x:\overline{\Omega}\rightarrow \mathbb{R}^n | x \in \mathcal{C}^0(\overline{\Omega}), x \in \mathcal{P}_p(T)^n , \forall T \in \mathcal{T}_h\},
\end{equation}
where $\mathcal{P}_p(T)$ is the space of functions that are polynomials of maximum degree $p\in\mathbb{N}_0$ on interval $T$. 
In collocation methods the dynamic constraints \eqref{eq:dynODE} are enforced at a finite number of collocation points in each $T_k$, the location and number of which depend on the chosen method. Denote collocation points in $T_k$ as $t^{c,k}_j \in \mathcal{T}_{k,c}$, and
\begin{equation}
\dot{\tilde{x}}(t^{c,k}) = f(\tilde{x}(t^{c,k}),u(t^{c,k}),t^{c,k}), \forall t^{c,k} \in \mathcal{T}_{k,c}, k \in \mathcal{K}.
\end{equation}

In $h$-methods (e.g.\ Euler, Hermite-Simpson), a fixed-degree polynomial is used on each segment \cite{betts2010practical}. For fixed~$p$, the approximation accuracy may be improved by reducing~$h$ and increasing the number of segments.
In $p$-methods the unknown trajectories are approximated as an interpolation of orthogonal basis functions \cite{fahroo2008advances}.
In global $p$-methods, a single interpolated function is used over the entire interval~$\Omega$.
Approximation accuracy is adjusted by changing the polynomial order.

Mesh refinement schemes iteratively adjust the number of segments and/or polynomial order of the transcription. At each iteration of the refinement, problem \eqref{eq:Bolza_cont} must be re-transcribed and the resulting NLP solved. Such schemes guarantee a user-defined level of accuracy of the approximation, which is important in applications.
We must therefore consider three different trajectories. The \emph{true} solution of \eqref{eq:Bolza_cont}, $x^*(\cdot),u^*(\cdot)$, which is in general unknown. The \emph{approximate} solution $\tilde{x}(\cdot),u(\cdot)$ which depends on the chosen transcription method and is computed numerically.
Finally, the \emph{actual} state trajectory $x(\cdot)$ of the plant resulting from evaluating \eqref{eq:sysDyn} with the approximate input $u(\cdot)$. In general $\tilde{x}(t)\not = x^*(t) \not = x(t),\forall t \not \in \mathcal{T}_u$. A pictorial representation of this is shown in Fig.~\ref{fig:TrajectoryApproxEx}.

\begin{figure}[tb]
	\centering
	\includegraphics[width=\columnwidth]{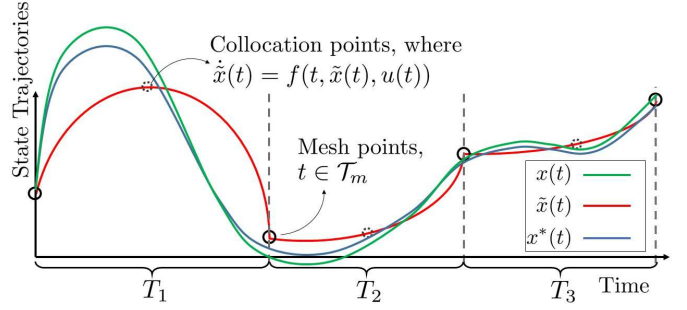}
	\caption[TODO]{A  comparison of $x(\cdot)$, $x^*(\cdot)$ and the approximate solution $\tilde{x}(\cdot) \in \mathcal{X}_{P}$ defined on mesh $\mathcal{T}_h$.}
	\label{fig:TrajectoryApproxEx}
\end{figure}

Since $x^*(\cdot),u^*(\cdot)$ are not known, they cannot be used to evaluate the approximation accuracy.
One possible metric for refinement is the \emph{absolute local error} at time $t$ in the $i^\text{th}$ state,
\begin{equation}\label{eq:AbsoluteLocalError}
\varepsilon_i(t) := \dot{\tilde{x}}_i(t) - f_i(t,\tilde{x}(t),u(t)).
\end{equation}
The vector $\varepsilon(t) := [\varepsilon_1(t),\hdots,\varepsilon_n(t)]'$  is the error in the solution of the differential equations \eqref{eq:dynODE} resulting from the chosen transcription method. $\varepsilon(t)$ is necessarily zero at collocation points $t \in \mathcal{T}_{k,c},\forall k \in \mathcal{K}_h$. The quadrature of the approximation error for state $i$ over mesh interval $T_k\in\mathcal{T}_m$ is defined as
\begin{equation}\label{eq:AbsoluteQuadrature}
    \eta_{k,i} :=\int_{T_k} \norm{\varepsilon_i(\tau)}\mathrm{d}\tau, 
\end{equation}
where $\norm{\cdot}$ is the scalar norm.
As in \cite[Ch. 4]{betts2010practical}, we may obtain a \emph{relative} measure of the quadrature of the local error,
\begin{equation} \label{eq:RLE}
    \epsilon_{k,i} := \frac{\eta_{k,i}}{(w_i + 1)} ,
\end{equation}
where the scaling weight,
\begin{equation} \label{eq:RelativeLocalErrorScaling}
    w_i := \max_{t\in\mathcal{T}_m}\left\{\norm{\dot{\tilde{x}}_i(t)}, \norm{\tilde{x}_i(t)}\right\},
\end{equation}
defines the maximum value of the $i^\text{th}$ state or its derivative at each of the mesh points.
Mesh refinement schemes guarantee an upper bound on the quadrature of the approximation error in each segment. Since these schemes are heuristic in nature, there is choice in the metric used.
It is suggested in \cite{betts2010practical} to use the \emph{maximum relative local error} over all $i$ state components in interval $T_k\in\mathcal{T}_m$
\begin{equation} \label{eq:RLEmax}
    \epsilon_{k} := \max_{i}\epsilon_{i,k}.
\end{equation}
However, many practical problems require different states to be solved to different accuracies. As such, we consider using mesh refinement schemes to enforce the conditions
\begin{equation}\label{eq:RelLocErrBound}
    \epsilon_{k,i} \leq \hat{\epsilon}_i, \forall k \in \mathcal{K},
\end{equation}
where $\hat{\epsilon}:= [\hat{\epsilon}_1,\hdots,\hat{\epsilon}_n]^\intercal$ is a vector of user-defined inaccuracy tolerances.
Another alternative is to bound the \emph{absolute local errors},
\begin{equation} \label{eq:AbsLocErrBound}
    \eta_{k,i} \leq \hat{\eta}_i, \forall k \in \mathcal{K},
\end{equation}
where  $\hat{\eta}:= [\hat{\eta}_1,\hdots,\hat{\eta}_n]^\intercal$ is again a vector of user-defined tolerances. Refinement schemes based on absolute local error may perform worse (in terms of number of iterations and the resulting NLP size). However, the tolerances $\hat{\eta}$ will often have a physical meaning, making them easier for control practitioners to determine. Although we assume a time-invariant tolerance for each state over the prediction horizon, the proposed framework and analysis is easily extended to allow for different tolerances at each mesh interval. This is used, for example, in \cite{paiva2017sampled} where the requirements on solution accuracy are reduced further in the prediction horizon.

\section{Minimum IUT} \label{sec:MinIUT_Offline}

Define the error $\delta(t,t_k^u)$ between the predicted state $\tilde{x}(t,t_k^u)$, based on state measurement $x(t_k^u)$, and measurement $x(t)$ as
\begin{equation}
    \delta(t,t_k^u) := \tilde{x}(t,t_k^u)-x(t).
\end{equation}
The components of the \emph{prediction error} in each state are $\delta_i(t) := \tilde{x}_i(t) - y_i(t)$, with $\delta(t):=[\delta_1(t),\hdots,\delta_n(t)]^\intercal$ where $\delta(t) = 0, \forall t \in \mathcal{T}_u$, with the associated prediction error dynamics $\dot{\delta}(t) := \dot{\tilde{x}}(t) - (\dot{x}(t) + \dot{\theta}(t))$. Consider employing the triggering condition
\begin{equation} \label{eq:EventTriggeringCondition}
    \tilde{F}_\text{event} := \Delta - \vecnorm{\delta(t)}_M.
\end{equation}
In the following analysis $\vecnorm{\cdot}_M$ is assumed to be an appropriately defined $p$--norm with weighting matrix $M$, defined as
\begin{equation}
\vecnorm{x}_M := \left(\sum_{i=1}^{n} M_i\norm{x_i}^p\right)^{\frac{1}{p}},
\end{equation}
where the matrix $M\in\mathbb{R}^{n\times n}$ is a positive diagonal weighting matrix. $M_i$ is the $i^\text{th}$ component of the diagonal of $M$. In the following, we also use the definition of an induced matrix norm
\begin{equation} \label{eq:matrixNormDefinition}
\vecnorm{M} := \max_{\vecnorm{x}=1} \vecnorm{M x}.
\end{equation}

Note that the results of this work may be easily extended to triggering conditions $\tilde{F}_\text{event}$ of the form
\begin{equation} \label{eq:trigGeneral}
\Delta - g(\vecnorm{\delta(t)}_M),
\end{equation}
for any monotonically strictly increasing $g:\mathbb{R}\rightarrow\mathbb{R}$. 

\begin{theorem}[Minimum IUT] \label{th:MinIUTRel}
	For the ETC scheme resulting from solving \eqref{eq:Bolza_cont} to an accuracy satisfying \eqref{eq:RelLocErrBound} at times $\mathcal{T}_u$ implicitly determined through triggering condition \eqref{eq:EventTriggeringCondition}, the minimum IUT
	\begin{equation}\label{eq:min_IUT}
	 \underline{\tau}_u := \frac{\Delta - \vecnorm{(\hat{w} + 1)\hat{\epsilon}}_M + \vecnorm{M}2\hat{\theta}}{\frac{\vecnorm{(\hat{w} + 1)\hat{\epsilon}}_M}{\sigma h} + \vecnorm{M}\left(L_x \Delta + \hat{v}\right)},
	\end{equation}
	where $w_i \leq \hat{w}_i$ and $\hat{w}:= [\hat{w}_1,\hdots,\hat{w}_n]$, is guaranteed to be strictly positive for $\Delta > \vecnorm{(\hat{w}+1)\hat{\epsilon}}_M + \vecnorm{M}2\hat{\theta}$.
\end{theorem}

\begin{proof}
	The ETC \eqref{eq:EventTriggeringCondition} guarantees that
	\begin{equation}
	\vecnorm{\delta(t)}_M \leq \Delta.
	\end{equation}
	Without loss of generality, assume problem \eqref{eq:Bolza_cont} was last solved at time $t_i^u=0$, parameterized by the measurement $\tilde{x}(0)=y(0)$. Then the prediction error in the $i^\text{th}$ state is
	\begin{align}
	\delta_i(t) &= \int_0^t \dot{\tilde{x}}_i(\tau) - x_i(\tau) \mathrm{d}t \\
	&= \int_0^t \dot{\tilde{x}}_i(\tau) - \left(f_i(\tau,x(\tau),u(\tau)) + v_i(\tau) + \dot{\theta}_i(\tau)\right)\mathrm{d}\tau \\
	&= \underbrace{\int_0^t \dot{\tilde{x}}_i(\tau) - f_i(\tau,x(\tau),u(\tau)) + v_i(\tau)\mathrm{d}\tau}_{:=\gamma_i(t)} + \theta_i(t)-\theta_i(0) 
	\end{align}
	Having isolated the measurement noise, we focus on the integral denoted by $\gamma_i(t)$. By noting that scalar $w_i$ is strictly positive and not a function of time, we may rewrite \eqref{eq:RLE} as
	\begin{equation} \label{eq:th1_1}
	\epsilon_{k,i} = \int_{t_k}^{t_{k+1}}\underbrace{\norm{\frac{\dot{\tilde{x}}_i(\tau) - f_i(\tau,\tilde{x}(\tau),u(\tau))}{w_i+1}}}_{:=\norm{\alpha_i(\tau)}}\mathrm{d}\tau.
	\end{equation}
	Subsequent substitution of $\dot{\tilde{x}}(t)= (w_i+1)\alpha(t) + f_i(x(t))$ into \eqref{eq:th1_1} results in
	\begin{equation}
	\gamma_i(t) = \int_0^t (w_i+1)\alpha_i(\tau) + f_i(\tau,\tilde{x}(\tau),u(\tau))-f_i(\tau,x(\tau),u(\tau)) - v_i(\tau)\mathrm{d}\tau
	\end{equation}
	Through application of the triangle inequality, the prediction error magnitude of each state may be bounded from above as
	\begin{align}
	\norm{\gamma_i(t)} &\leq \norm{(w_i+1)\int_0^t\alpha_i(\tau)\mathrm{d}\tau} \notag + \\ & \qquad \norm{\int_0^t f_i(\tau,\tilde{x}(\tau),u(\tau))-f_i(\tau,x(\tau),u(\tau))\mathrm{d}\tau} + \\ & \qquad \quad \norm{\int_0^t -v_i(\tau)\mathrm{d}\tau} \\
	&\leq (w_i+1)\int_0^t\norm{\alpha_i(\tau)}\mathrm{d}\tau + \\ & \qquad \int_0^t \norm{f_i(\tau,\tilde{x}(\tau),u(\tau))-f_i(\tau,x(\tau),u(\tau))} \mathrm{d}\tau + \\ & \qquad \quad \int_0^t \norm{v_i(\tau)}\mathrm{d}\tau \label{eq:ThreeIntegrals_1}.
	\end{align}
	Now, consider the set $\mathcal{T}_{t} \subseteq \mathcal{T}_h$ of open intervals of $T$, i.e.
	\begin{equation}
	\mathcal{T}_{t} := \{T \in \mathcal{T}_h | T \cap \overline{(0,t)} \not = \emptyset\}
	\end{equation}
	which is a mesh on an interval contained in $\Omega$ and by construction satisfies the properties detailed in Definition~\ref{def:Mesh}. Since the mesh refinement scheme guarantees $\epsilon_{k,i} \leq \hat{\epsilon}_i, \forall i \in \{1,\hdots,n\}$, the first term in \eqref{eq:ThreeIntegrals_1} may be bounded from above as
	\begin{equation}
	(w_i + 1)\int_0^t\norm{\alpha_i(\tau)}\mathrm{d}\tau \leq (w_i + 1)\sum_{T \in \mathcal{T}_{t}} \int_T \norm{\alpha_i(\tau)}\mathrm{d}\tau \leq (w_i+1) \sum_{T \in \mathcal{T}_{t}} \hat{\epsilon}.
	\end{equation}
	Through application of the resolution and quasi-uniformity properties of mesh $\mathcal{T}_t$, the smallest interval satisfies
	\begin{equation}
	\min_{T\in\mathcal{T}_t} \norm{T} \geq \sigma h := \underline{h}.
	\end{equation}
	Using $\ceil{\cdot}$ to denote the ceiling function, we may bound the maximum number of mesh segments in the interval $(0,t)$ as $\ceil{t\underline{h}^{-1}}$, from which it follows that
	\begin{equation}
	\sum_{T \in \mathcal{T}_{t}} \hat{\epsilon}_i \leq \ceil*{\frac{t}{\underline{h}}} \hat{\epsilon}_i \leq \left(\frac{t}{\underline{h}}+1\right) \hat{\epsilon}_i.
	\end{equation}
	From this, we can bound the prediction error in the $i^\text{th}$ state as
	\begin{align}
	\norm{\delta_i(t)} &\leq 
	\left(\frac{t}{\underline{h}}+1\right)(w_i + 1)\hat{\epsilon}_i + \norm{\theta_i(t)-\theta_i(0)} + \\& \
	\int_0^t\norm{f_i(\tau,\tilde{x}(\tau),u(\tau)) - f_i(\tau,x(\tau),u(\tau))}\mathrm{d}\tau + \int_0^t \norm{v_i(\tau)}\mathrm{d}\tau\\
	&\leq \left(\frac{t}{\underline{h}}+1\right) (w_i + 1)\hat{\epsilon}_i + \norm{\theta_i(t)-\theta_i(0)} + \\& \
	L_x\int_0^t\norm{\delta_i(\tau)}\mathrm{d}\tau + \int_0^t \norm{v_i(\tau)}\mathrm{d}\tau
	\end{align}
	where the second inequality results from the assumption of Lipschitz continuity of $f(\cdot)$.
	Combining each element-wise inequality, we can bound the weighted norm of the prediction error as
	\begin{align}
	\vecnorm{\delta(t)}_M &\leq \Bigg( \sum_{i=1}^n M_i \Bigg(
	\left(\frac{t}{\underline{h}}+1\right) (\hat{w} + 1)\hat{\epsilon}_i + \norm{\theta_i(t)-\theta_i(0)} + \\ & \qquad
	L_x\int_0^t\norm{\delta_i(\tau)}\mathrm{d}\tau + \int_0^t \norm{v_i(\tau)}\mathrm{d}\tau
	\Bigg)^p\Bigg)^{\frac{1}{p}}. 
	\end{align}
	Through application of the triangle inequality we may then bound the prediction error at time $t$ as, 
	\begin{align}
	\vecnorm{\delta(t)}_M  &\leq 
	\left(\frac{t}{\underline{h}}+1\right) \vecnorm{(\hat{w} + 1)\hat{\epsilon}}_M + \vecnorm{\theta(t)-\theta(0))}_M + \\ &
	\qquad
	L_x\int_0^t \vecnorm{\delta(\tau)}_M \mathrm{d}\tau + \int_0^t \vecnorm{v(\tau)}_M\mathrm{d}\tau \label{eq:PreErrorNorm1}
	 \\
	&\leq \left(\frac{t}{\underline{h}} + 1 \right) \vecnorm{(\hat{w} + 1)\hat{\epsilon}}_M + \vecnorm{M}\left(  2\hat{\theta} + t \left(L_x\Delta + \hat{v}\right)\right).
	\end{align}
	The second inequality follows from
	\begin{equation}
	\vecnorm{x}_M  \leq \vecnorm{M}\vecnorm{x}.
	\end{equation}
	We have also made use of the assumptions that $\vecnorm{\theta(t)-\theta(0)} \leq 2\hat{\theta}$ and $\vecnorm{v(t)}\leq \overline{v}$. 
	We now have an expression which is dependent only on parameters known offline. We may guarantee the existence of a $\underline{\tau}^u > 0$ by showing the existence of a $t>0$ such that   
	\begin{equation}
	\left(\frac{t}{\underline{h}} + 1 \right) \vecnorm{(\hat{w} + 1)\hat{\epsilon}}_M + \vecnorm{M}\left(  2\hat{\theta} + t \left(L_x\Delta + \hat{v}\right)\right) \leq \Delta.
	\end{equation}
	Through manipulation, this inequality is satisfied when,
	\begin{equation}
	t \leq \frac{\Delta - \vecnorm{(\hat{w} + 1)\hat{\epsilon}}_M + \vecnorm{M}2\hat{\theta}}{\frac{\vecnorm{(\hat{w} + 1)\hat{\epsilon}}_M}{\sigma h} + \vecnorm{M}\left(L_x \Delta + \hat{v}\right)}.
	\end{equation}
	Since all parameters are in $\mathbb{R}_0^+$, the RHS is guaranteed to be strictly positive as long as $\Delta > \vecnorm{(\hat{w}+1)\hat{\epsilon}}_M + \vecnorm{M}2\hat{\theta}$.
\end{proof}

In the special case, detailed in \cite[Ch. 4]{betts2010practical}, where refinement occurs based only on the maximum relative error,
guaranteeing
\begin{equation}
\max_i \epsilon_{k,i} \leq \epsilon_\text{max}, \forall k \in \mathcal{K} \Rightarrow  \epsilon_{k,i} \leq \hat{\epsilon}, \forall i \in \{1,\hdots,n\},  k \in \mathcal{K},
\end{equation}
then the resulting guaranteed IUT may be written as
\begin{equation}
	\underline{\tau}^u := \frac{\Delta - \vecnorm{\hat{w}+1}_M\epsilon_\text{max} - \vecnorm{M}2\hat{\theta}}{\frac{\vecnorm{\hat{w}+1}_M\epsilon_\text{max}}{\sigma h} + \vecnorm{M}\left(L_x \Delta + \hat{v}\right)}.
\end{equation}

In a basic mesh refinement scheme, the NLP resulting from \eqref{eq:Bolza_cont} is solved on each new, increasingly dense mesh. In \cite{betts2010practical} it is suggested to  use an initially sparse mesh, where the refinement adds grid points where necessary. Therefore, a small $\hat{\eta}$ invariably results in more NLPs being solved for a single approximation $\tilde{x}(\cdot),u(\cdot)$. 
Since \eqref{eq:min_IUT} is a decreasing function of $\hat{\epsilon}$, we can conclude a trade off between solving more NLPs per control update --- yielding a higher accuracy solution --- but less frequently, so solving fewer NLPs at each update, but more frequently.

We refine the mesh of a direct collocation transcription using local error analysis. The driving principle of direct collocation is that local changes of the mesh or polynomial order will only affect local changes in the solution. Hence, although the above trade-off \emph{is} true for any refinement scheme, we can also affect this trade-off by tailoring the operation of the refinement procedure used. The work  \cite{betts2010practical} proposes using only the current error distribution of the primal variables as a means of estimating how the error may be reduced by adding additional mesh points in specific intervals. On the other hand, \cite{paiva2017sampled} refines using a metric based on the error distribution of the dual variables --- this maintains important properties as a representation of solution quality while being computationally cheaper. 
Regardless, adding too many mesh points increases the dimensionality of the following NLP unnecessarily. By contrast, adding too few mesh points results in more refinement iterations (and NLP solves) needed to achieve allowable local errors.

In this qualitative analysis we have also not considered the effect of warm starting the NLP solver (both at each mesh refinement iteration and each update time) with the previous solution. 
We have also not considered mesh refinement schemes that may also \emph{remove} grid points.
Both of these considerations will be useful for reducing the computational cost of online schemes.

Unfortunately the bound \eqref{eq:min_IUT} is not tight. 
For the update to be triggered at time $t_k^u + \underline{\tau}^u$, the approximation error metric would need to be maximized $\forall t \in [t_k^u,t_k^u + \underline{\tau}^u)$. However, this cannot be the case, because the refinement metric is reset to zero at time $t_k^u$. Since the dynamics are assumed smooth and approximating polynomials piecewise smooth, the refinement metric cannot be maximized immediately.
Another important point is that bound \eqref{eq:min_IUT} will be finite even when $\hat{v}=0$ and $\hat{\theta}=0$. We have not used knowledge of the discretization/interpolation schemes, in which case the bound may be improved.

If \eqref{eq:pathConstriants} includes simple bounds and rate constraints on all state variables we may guarantee a finite $\hat{w}$. For example, the bounds $\underline{x} \leq x_i(t) \leq \overline{x}$ and $ \underline{\dot{x}} \leq \dot{x}_i(t) \leq \overline{\dot{x}} $ imply that
\begin{equation} \label{eq:maxScaling1}
   w_i \leq \max\Big\{\norm{\underline{x}},\norm{\overline{x}},\norm{\underline{\dot{x}}},\norm{\overline{\dot{x}}} \Big\},
\end{equation}
from which $\hat{w}$ may be determined. Similarly, if we make the additional assumption that the origin is an equilibrium --- common in some regulation-type problems --- then
\begin{equation} \label{eq:maxScaling2}
    w_i \leq \max\Big\{\norm{\underline{x}},\norm{\overline{x}},L_x\norm{\underline{x}},L_x\norm{\overline{x}} \Big\},
 \end{equation}
without needing to bound the state derivatives.

However, in certain cases it may not be natural or appropriate to enforce these constraints. Further, in non-regulation type problems, $x(t)$ may be significantly within bounds, and so determination of $\hat{w}$ through \eqref{eq:maxScaling1} or \eqref{eq:maxScaling2} could be extremely restrictive. Thus, in the general case it is difficult to quantitatively relate the relative local quadrature~$\epsilon_k$ with the absolute quadrature $\eta_k$. Instead, refining directly on the absolute error yields the following result.
\begin{corollary} \label{cor:absoluteMinIUT}
    For the ETC scheme resulting from solving \eqref{eq:Bolza_cont} to an accuracy satisfying \eqref{eq:AbsLocErrBound} at times $\mathcal{T}_u$ implicitly determined through triggering condition \eqref{eq:EventTriggeringCondition}, the minimum IUT
    \begin{equation}\label{eq:min_IUT_Absolute}
        \underline{\tau}_u := \frac{\Delta - \vecnorm{\hat{\eta}}_M + \vecnorm{M}2\hat{\theta}}{\frac{\vecnorm{\hat{\eta}}_M}{\sigma h} + \vecnorm{M}\left(L_x \Delta + \hat{v}\right)}
    \end{equation}
    is guaranteed to be strictly positive for $\Delta > \vecnorm{\hat{\eta}}_M + \vecnorm{M}2\hat{\theta}$.
\end{corollary}
\begin{proof}
   Follows the procedure for the proof of Theorem~\ref{th:MinIUTRel}, but with all $w_i=0$ and replacing $\epsilon$ by $\eta$.
\end{proof}

\section{Online estimation of IUT} \label{sec:OnlineTrig}
The IUT bound \eqref{eq:min_IUT} can be improved by using the actual error information provided by solving \eqref{eq:Bolza_cont} with a mesh refinement procedure. 
Theorem~\ref{th:QET} is the first explicit definition of IUT to account for both problem data and the solution accuracy of \eqref{eq:Bolza_cont}.
Although the following results resemble STC conditions of the form \eqref{eq:STC_Prob_Condition}, we make no claims about stability or performance of control based on these conditions.
Additional assumptions on problem structure would be needed for this.
Instead, we simply provide estimates of IUTs, which may be useful for process/resource scheduling. 
  
First, we determine a lower bound on $\tau_i^u$, not dependent on~$\eta_k$, through using knowledge of the collocation constraints. This may be applicable in the case where we cannot guarantee a solution accuracy, such as if only a single NLP solve is performed.

\begin{theorem}[Collocation Triggering (CT)] \label{th:CollTrig}
	Let $L_{p,k}$ be the Lipschitz constant of the approximating polynomial  of the state used in segment $T_k$. The IUT
	\begin{equation} \label{eq:trig_CT}
	\begin{split}
	& \tau_i^\text{CT} = \sup \Bigg\{\tau_i \in \mathbb{R}^+ | \\ & \qquad  	\Bigg(\vecnorm{
		\sum_{k\in\mathcal{K}_t}\int_{T_k}  L_{p,k}\frac{\norm{\tau - t_k^\text{col}}}{(w + 1)} + L_x \frac{\norm{\tilde{x}_i(t_k^\text{col}) - \tilde{x}_i(\tau)}}{(w+1)}\mathrm{d}\tau}_M +
	\\ & \qquad \vecnorm{M}(2\hat{\theta} + \tau_i \hat{v}) \Bigg)e^{L_x\tau_i} \leq \Delta \Bigg\}
	\end{split}
	\end{equation}
	calculated from the approximate solution $\tilde{x}(\cdot)$ to problem~\eqref{eq:Bolza_cont} at time $t_i^u$ satisfies $\vecnorm{\delta(t_i^u+\tau_i)}_M \leq \Delta, \forall \tau \leq \tau_i^\text{CT}$.
\end{theorem}
\begin{proof}
	The proof is similar to that of Theorem~\ref{th:MinIUTRel}.
	We begin by noting that
	\begin{equation} \label{eq:CT1}
	\begin{split}
	\int_0^t  &\norm{\dot{\tilde{x}}_i(\tau) - f_i(\tau,\tilde{x}(\tau),u(\tau))}\mathrm{d}\tau  \leq \\ & \qquad \sum_{k\in\mathcal{K}_t}\int_{T_k}\norm{\dot{\tilde{x}}_i(t) - f_i(\tau,\tilde{x}(\tau),u(\tau))} \mathrm{d}\tau,
	\end{split}
	\end{equation}
	where the inequality follows from $\sup{T_k} \geq t$. Since \eqref{eq:dynODE} is satisfied exactly at points $t_k^\text{col}$, we may write \eqref{eq:CT1} as
	\begin{align}
	&\sum_{k\in\mathcal{K}_t}\int_{T_k} \Big|\dot{\tilde{x}}_i(\tau) - \dot{\tilde{x}}_i(t_k^\text{col}) + \\ & \qquad f_i(t_k^\text{col}, \tilde{x}(t_k^\text{col}), u(t_k^\text{col}))- f_i(\tau,\tilde{x}(\tau),u(\tau))\Big|\mathrm{d}\tau \\
	&\leq \sum_{k\in\mathcal{K}_t}\int_{T_k} \norm{\dot{\tilde{x}}_i(\tau) - \dot{\tilde{x}}_i(t_k^\text{col})} + \\& \qquad \norm{f_i(t_k^\text{col}, \tilde{x}(t_k^\text{col}), u(t_k^\text{col}))- f_i(\tau,\tilde{x}(\tau),u(\tau))}\mathrm{d}\tau.
	\end{align}
	Therefore, the relative error quadrature in the first $\mathcal{T}_t$ segments is bounded as
	\begin{align}
	\sum_{T\in\mathcal{T}_t}\epsilon_{k,i} &\leq \frac{1}{(w_i + 1)}\sum_{k\in\mathcal{K}_t}\int_{T_k} \norm{\dot{\tilde{x}}_i(\tau) - \dot{\tilde{x}}_i(t_k^\text{col})} + \\ & \qquad \norm{f_i(t_k^\text{col}, \tilde{x}(t_k^\text{col}), u(t_k^\text{col}))- f_i(\tau,\tilde{x}(\tau),u(\tau))}\mathrm{d}\tau \\
	&\leq \frac{1}{(w_i + 1)}\sum_{k\in\mathcal{K}_t}\int_{T_k} L_{p,k}\norm{\tau - t_k^\text{col}} + \\ & \qquad L_x\norm{\tilde{x}_i(t_k^\text{col}) - \tilde{x}_i(\tau)}\mathrm{d}\tau 
	\end{align}
	where the second inequality results from noting that all polynomials $\in \mathcal{P}_p(T_k)$ defined on bounded set $T_k$ are Lipschitz functions, and using $L_{p,k}$ to denote the Lipschitz constant of the approximating polynomial used in segment~$T_k$.
	Summing this error expression for each  state results in
	\begin{equation}
	\vecnorm{\sum_{T\in\mathcal{T}_t}\epsilon_k}_M  \leq \vecnorm{
	\sum_{k\in\mathcal{K}_t}\int_{T_k}  L_{p,k}\frac{\norm{\tau - t_k^\text{col}}}{(w + 1)} + L_x\frac{\norm{\tilde{x}_i(t_k^\text{col}) - \tilde{x}_i(\tau)}}{(w+1)}\mathrm{d}\tau
	}_M.
	\end{equation}
	Substituting this quadrature error bound into \eqref{eq:PreErrorNorm1} in place of $(t\underline{h}^{-1}+1)\vecnorm{(\hat{w}+1)\hat{\epsilon}}$ results in
	\begin{align}
	\vecnorm{\delta(t)}_M  &\leq 
 	\vecnorm{
 	\sum_{k\in\mathcal{K}_t}\int_{T_k}  L_{p,k}\frac{\norm{\tau - t_k^\text{col}}}{(w + 1)} + L_x\frac{\norm{\tilde{x}_i(t_k^\text{col}) - \tilde{x}_i(\tau)}}{(w+1)}\mathrm{d}\tau
 	}_M  \\&
 	+ \vecnorm{M}(2\hat{\theta} + t \hat{v}) + 
	L_x\int_0^t \vecnorm{\delta(\tau)}_M \mathrm{d}\tau
	\end{align} 
	Using the integral form of the Gronwall-Bellman inequality we find
	\begin{align}
	\vecnorm{\delta(t)}_M \leq &
	\Bigg(\vecnorm{
		\sum_{k\in\mathcal{K}_t}\int_{T_k}  L_{p,k}\frac{\norm{\tau - t_k^\text{col}}}{(w + 1)} + L_x\frac{\norm{\tilde{x}_i(t_k^\text{col}) - \tilde{x}_i(\tau)}}{(w+1)}\mathrm{d}\tau
	}_M \\ & \qquad
	+ \vecnorm{M}(2\hat{\theta} + t \hat{v}) \Bigg)e^{L_xt}
	\end{align} 
	Restricting the RHS to be $ \leq \Delta$ yields the explicitly defined IUT given in \eqref{eq:trig_CT}.
\end{proof}

For low-order $h$ methods in particular, $L_{p,k}$ may be easily evaluated for each segment $T_k$. However, as we have no guarantees on the accuracy of the solution resulting from these polynomials, we cannot guarantee that $\tau_i^\text{CT}$ will be greater than $\underline{\tau}_u$ from Theorem~\ref{th:MinIUTRel}.

\begin{theorem}[Quadrature Error Triggering (QET)]\label{th:QET}
	The IUT
	\begin{equation}\label{eq:trig_QET}
	\tau_i^\text{QET} = \sup  \Bigg\{\tau >0  | \left(\vecnorm{\sum_{k\in\mathcal{K}_t}\epsilon_k}_M +  \vecnorm{M}\left(2\hat{\theta}  +\hat{v}\tau\right)\right) e^{ L_x\tau} \leq \Delta  \Bigg\}
	\end{equation}
	determined from $\tilde{x}(\cdot)$, $\epsilon_k$, which approximately solves problem \eqref{eq:Bolza_cont} at time $t_i^u$, satisfies $\vecnorm{\delta(t_i^u+\tau)}_M \leq \Delta, \forall \tau \leq \tau_i^\text{QET}$. Furthermore, $\tau_i^\text{QET} \geq \underline{\tau}^u$.
\end{theorem}
\begin{proof}
	Employing the notation introduced in the proof of Theorem~\ref{th:MinIUTRel}, we begin by noting that
	\begin{equation}
	(w_i+1) \int_0^t\norm{\alpha_i(\tau)}\mathrm{d}\tau \leq (w_i+1)\sum_{k \in \mathcal{K}_t} \epsilon_{k,i},
	\end{equation}
	where the inequality follows from $\sup \norm{\mathcal{T}_t} \geq t$. This is similar to the steps taken in Theorem~\ref{th:MinIUTRel}, except we use the actual error quadratures associated with the problem instead of the bounds $\hat{\epsilon}$.
	Substituting the actual quadrature errors into \eqref{eq:PreErrorNorm1} in place of $(t\underline{h}^{-1}+1)\vecnorm{(\hat{w}+1)\hat{\epsilon}}$ results in
	\begin{equation}
	\vecnorm{\delta(t)}_M \leq 
	\vecnorm{ (w+1)\sum_{k\in\mathcal{K}_t}\epsilon_k}_M	+ \vecnorm{M}(2\hat{\theta} + t \hat{v}) + 
	L_x\int_0^t \vecnorm{\delta(\tau)}_M \mathrm{d}\tau
	\end{equation}
	We again use the integral form of the Gronwall-Bellman inequality to find
	\begin{equation} \label{eq:PreErrorNorm2}
	\vecnorm{\delta(t)}_M \leq 
	\left(
	\vecnorm{ (w+1)\sum_{k\in\mathcal{K}_t}\epsilon_k}_M	+ \vecnorm{M}(2\hat{\theta} + t \hat{v}) \right)e^{L_xt}
	\end{equation} 
	Restricting the RHS to be $ \leq \Delta$ yields the explicitly defined IUT given in \eqref{eq:trig_CT}.

	Restricting the RHS to be $ \leq \Delta$ results in the explicitly defined IUT given in \eqref{eq:trig_QET}. As
	\begin{equation}
	\sum_{i=1}^n (w_i+1)\sum_{k \in \mathcal{K}_t} \epsilon_{k,i} \leq  \left(\frac{t}{\underline{h}}+1\right)\hat{\epsilon}_i,
	\end{equation}
	we guarantee that $\tau^\text{QET}_i \geq \underline{\tau}^u$.
\end{proof}

Another advantage of the online calculation is that it does not require explicit knowledge of bounds on $w$. 
In the same fashion as Corollary~\ref{cor:absoluteMinIUT}, we may determine a similar triggering estimate if considering refinement on the absolute local error of each state.

The mesh refinement procedure may not return the approximation error over arbitrary intervals. Software such as \texttt{ICLOCS2} \cite{nie2018iclocs2} returns $\eta_k$ or $\epsilon_k$ corresponding to $T_k$. Depending on how the integral is evaluated we may not know the distribution of $\varepsilon(t)$ over this interval. 

Up to this point the analysis has made no assumptions on problem structure. As such it is equally applicable to the different classes of event-triggered MPC formulations proposed in the literature which use triggering conditions of the form \eqref{eq:trigGeneral}.

\section{Constraint tightening} \label{sec:constraintTightening}

The proposed triggering conditions do not account for state proximity to constraint bounds. We wish to guarantee satisfaction of constraints \eqref{eq:pathConstriants} for the entire prediction horizon $t \in [t_0, t_f]$. 
However, when solving the transcribed NLP corresponding to problem \eqref{eq:Bolza_cont}, the constraints $c(\tilde{x}(t),u(t),t) \leq 0$ are only enforced on the approximate trajectory $\tilde{x}(\cdot)$ at a finite number of time instances: the mesh points $\mathcal{T}_m$. 
One way to guarantee that the actual state $x(\cdot)$ complies with state constraints is to employ constraint tightening methods. 
For the following analysis we only consider the case where $c(\tilde{x}(t),u(t),t) \leq 0$ is equivalent to the inclusion
\begin{equation}
    \tilde{x}(t) \in \mathbb{X}, \forall t \in [t_0,t_f],
\end{equation}
where the constraint set $\mathbb{X}$ satisfies the following assumption.
\begin{assumption}\label{ass:Nonempty}
The set $\mathbb{X}$ is closed and convex, and the reduced set
\begin{equation}
    \mathbb{X} \ominus \left\{x \in \mathbb{R}^n \mid x \leq \Delta \right\},
\end{equation}
where $\ominus$ denotes the Pontryagin difference, is non-empty.
\end{assumption}

The Pontryagin difference of two sets $\mathcal{A}$ and $\mathcal{B}$ is defined as $\mathcal{A} \ominus \mathcal{B} := \{a \in \mathbb{R}^n | a+b \in \mathcal{A}, \forall b \in \mathcal{B}\}$.
The Minkowski sum $\oplus$ of $\mathcal{A}$ and $\mathcal{B}$ is $\mathcal{A} \oplus \mathcal{B} := \{a + b | a \in \mathcal{A}, b \in \mathcal{B}\}$. We would not expect this assumption to be particularly restrictive, since $\Delta$ will be quite small in practice.
In the following result we evaluate by how much the state constraints on the approximate solution $\tilde{x}$ must be tightened in order to guarantee that the actual state $x(t) \in \mathbb{X}$.

\begin{lemma}\label{lemma:ConstraintTightening}
    Consider a mesh refinement scheme using the relative local error \eqref{eq:RelLocErrBound} or absolute local error \eqref{eq:AbsoluteQuadrature} as a termination criteria.
    If the approximate trajectory at time $t_k^u+s, s \geq 0,$ predicted from measurement at time $t_k^u$ satisfies $\tilde{x}(t_k^u,s) \in \mathbb{X}_{s},$ where
    \begin{equation} \label{eq:infTightening}
        \mathbb{X}_{s} := \mathbb{X} \ominus \mathcal{B}_{s},
    \end{equation}
    and the ball $\mathcal{B}_s$ is defined as
	\begin{equation}
	\begin{split}
	\mathcal{B}_{s} := & \Bigg\{x \in \mathbb{R}^n \mid  \\& \qquad
	\vecnorm{x}_M \leq \min\Big\{  \left(
	\vecnorm{ \sum_{k\in\mathcal{K}_s}\hat{\xi}}_M	+ \vecnorm{M}(2\hat{\theta} + s \hat{v}) \right)e^{L_xs}, \Delta \Big\}\Bigg\},
	\end{split}
	\end{equation}
	where $\hat{\xi}:=\left(\hat{w}+1\right)\hat{\epsilon}$ or $\hat{\xi}:=\hat{\eta}$ depending on the refinement termination criteria, then the actual state trajectory satisfies $x(s) \in \mathbb{X}$.
\end{lemma}

\begin{proof}
	Take the inequality \eqref{eq:PreErrorNorm2} used in the proof of Theorem~\ref{th:QET}, but using error bounds instead of actual errors. We directly bound the norm between the approximate predicted and measured state at time $t_k^u+s$ by
	\begin{equation}
	\vecnorm{\delta(t_k^u+s)}_M \leq  \left(
	\vecnorm{ \sum_{k\in\mathcal{K}_s}\hat{\xi}}_M	+ \vecnorm{M}(2\hat{\theta} + s \hat{v}) \right)e^{L_xs}.
	\end{equation}
	Since the ETC scheme guarantees that the control and associated prediction is updated if $\vecnorm{\delta(s)} \geq \Delta$, we can make the bound less conservative, to become 
	\begin{equation}
	\vecnorm{\delta(t_k^u+s)}_M \leq \min\left\{\left(
	\vecnorm{ \sum_{k\in\mathcal{K}_s}\hat{\xi}}_M	+ \vecnorm{M}(2\hat{\theta} + s \hat{v}) \right)e^{L_xs}, \Delta\right\}.
	\end{equation}
	Tightening the constraint set $\mathbb{X}$ by the RHS of this inequality results in $\mathbb{X}_s$ and guarantees that the actual trajectory $x(t) \in \mathbb{X}, \forall t \in [t_k^u,t_{k+1}^u]$.
	By Assumption~\ref{ass:Nonempty}, $\mathbb{X}_{s}$ is nonempty $\forall s \in [t_k^u,t_k^u+T_\text{op})$.
\end{proof}

Lemma~\ref{lemma:ConstraintTightening} shows that the amount of constraint tightening required depends on the solution accuracy, allowed prediction error $\Delta$, as well as measurement and disturbance noise bounds. This tightening is shown in Figure~\ref{fig:ConstraintTighteningExample}. 
However, the constraint $\tilde{x}(t) \in \mathbb{X}_s$ would still need to be applied at the infinite number of time instances $[t_k^u,t_k^u + T_\text{op}]$. 
We now adapt \cite[Theorem 1]{fontes2018guaranteed} to show how further tightening of the constraint sets $\mathbb{X}_{s}$ at the mesh points $\mathcal{T}_m$ can be used to guarantee constraint satisfaction of the continuous-time OCP.

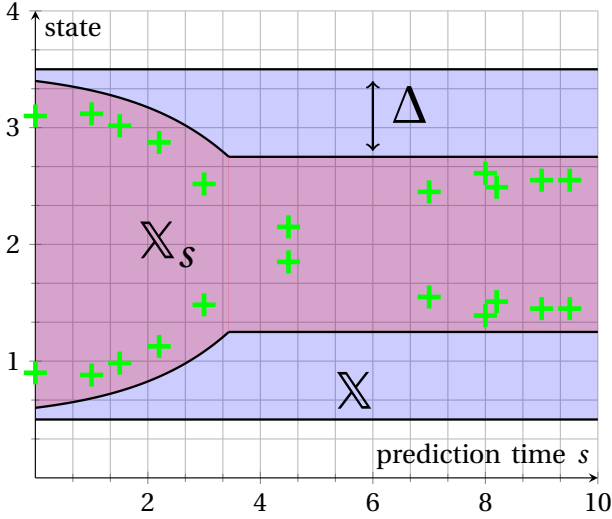
\begin{figure}[tb]
  \centering
  \resizebox{0.95\columnwidth}{!}{%
    \begin{tikzpicture} 
    \begin{axis}[
        xmin=0, xmax=10,    
        ymin=0, ymax=4,
        axis lines=middle,
         grid=both,xlabel={prediction time $s$},ylabel={state},minor tick num=2
      ]  
      \addplot[samples=100, domain=0:10, name path=A, style={thick}] {3.5}; 
      \addplot[samples=50, domain=0:10,name path=B, style={thick}] {0.5};
      \addplot[samples=50, domain=3.4387:10,name path=C, style={thick}] {2.75};
      \addplot[samples=50, domain=3.4387:10,name path=D, style={thick}] {1.25};
      \addplot[samples=50, domain=0:3.4387,name  path=E, style={thick}] {3.5 - (0.1 + 0.01*x)*pow(e,0.5*x)};
	  \addplot[samples=50, domain=0:3.4387,name  path=F, style={thick}] {0.5 + (0.1 + 0.01*x)*pow(e,0.5*x)};
      \path[name path=xaxis] (\pgfkeysvalueof{/pgfplots/xmin}, 0) -- (\pgfkeysvalueof{/pgfplots/xmax},0);
      \addplot[blue, opacity=0.2] fill between[of=A and B, soft clip={domain=0:10}];
      \addplot[red, opacity=0.2] fill between[of=C and D, soft clip={domain=3:10}];
      \addplot[red, opacity=0.2] fill between[of=E and F, soft clip={domain=0:4}];
	  \draw[<->, style={thick}](6,2.8)--(6,3.4);
	  \node[text width=3cm] at (8.5,3.2) {\huge $\Delta$};
	  \node[text width=3cm] at (7.5,0.75) {\huge $\mathbb{X}$};
	  \node[text width=3cm] at (4,2) {\huge $\mathbb{X}_s$};
	  \addplot[green, mark=+, only marks,mark size=4pt, ultra thick] coordinates{
        (0,	0.900000000000000)
        (1,	0.881359339777014)
        (1.5,	0.983455001910458)
        (2.2,	1.12650825492147)
        (3,	1.48261957914395)
        (4.5,	1.85000000000000)
        (7,	1.55000000000000)
        (8,	1.39000000000000)
        (8.2,	1.51000000000000)
        (9,	1.45000000000000)
        (9.5,	1.45000000000000)
   };
   	  \addplot[green, mark=+, only marks,mark size=4pt, ultra thick] coordinates{
        (0,	3.10000000000000)
        (1,	3.11864066022299)
        (1.5,	3.01654499808954)
        (2.2,	2.87349174507854)
        (3,	2.51738042085605)
        (4.5,	2.15000000000000)
        (7,	2.45000000000000)
        (8,	2.61000000000000)
        (8.2,	2.49000000000000)
        (9,	2.55000000000000)
        (9.5,	2.55000000000000)
   };
	  \end{axis}
 \end{tikzpicture}
 }
\caption[TODO]{A representation of the proposed constraint tightening method for the constraint set $\mathbb{X}:= [\underline{x},\overline{x}]$. $\mathbb{X}$ is tightened to the set $\mathbb{X}_s$ due to the difference between approximate and actual trajectories. The constraints enforced at discrete mesh points are further tightened depending on the length of the succeeding segment, shown by the green crosses.}
\label{fig:ConstraintTighteningExample}
\end{figure}

Similar to the method proposed in \cite{fontes2018guaranteed}, we define the signed distance function from a point $x\in\mathbb{R}^n$ to the set $\mathcal{A}$ as
\begin{equation}
    d_{\mathcal{A}}(x) := \min_{y\in\mathcal{A}}\vecnorm{x-y} - \inf_{z\in \mathcal{A}^C}\vecnorm{x-z},
\end{equation}
where $\mathcal{A}^C$ denotes the complement of $\mathcal{A}$.
With this distance function we may rewrite the tightened state constraints from Lemma~\ref{lemma:ConstraintTightening} as
\begin{equation}
    d_{\mathbb{X}_s}(\tilde{x}(t_k^u,s)) \leq 0, \forall s \in [t_k^u,t_k^u + T_\text{op}].
\end{equation}
An important advantage of using this set-based approach with the distance operator, compared to the standard inequality constraint formulation, is that we can normalize multiple different constraints and then treat them all in a uniform fashion --- as exemplified in the application in \cite{fontes2018guaranteed}.
The next step is to tighten the constraints on the approximate system at the finite set of mesh points $\mathcal{T}_m$, in order to guarantee constraint satisfaction of the real system over the entire time horizon. This is formalized in the following theorem.

\begin{theorem} \label{th:tight}
	Consider a mesh refinement scheme using the relative local error \eqref{eq:RelLocErrBound} or absolute local error \eqref{eq:AbsoluteQuadrature} as a termination criterion. If we define
	\begin{equation}
	\alpha_k \geq \vecnorm{\hat{\xi}} + L_f\norm{T_k} + g(t_k^m,t_{k+1}^m)
	\end{equation}
	and the predicted trajectory $\tilde{x}(\cdot)$ satisfies
	\begin{equation}
	d_{\mathbb{X}_{t_k^m}}(\tilde{x}(t_k^m)) \leq -\alpha_k, \forall k \in \mathcal{K}_h,
	\end{equation}
	then
	\begin{equation}
	d_{\mathbb{X}_{s}}(\tilde{x}(s)) \leq 0, \forall s \in [t_k^m,t_{k+1}^m], k \in \mathcal{K}_h
	\end{equation}
	and 
	\begin{equation}
	x(t) \in \mathbb{X}, \forall s \in [t_i^u, t_{i+1}^u].
	\end{equation}
\end{theorem}
\begin{proof}
	From the definition of the tightened constraint set $\mathbb{X}_{s}$ in Lemma~\ref{lemma:ConstraintTightening},
	\begin{equation}
	\mathbb{X}_{s} \subseteq \mathbb{X}_{t_k^m}, \forall s \in (t_k^m,t_{k+1}^m].
	\end{equation}
	We use the notation $\partial \mathcal{A}$ to denote the boundary of the set $\mathcal{A}$. 
	Since $\mathbb{X}$ is convex by Assumption~\ref{ass:Nonempty}, the sets $\mathbb{X}_s$ will be convex by construction. Therefore the shortest distance from any point in $\partial\mathbb{X}_s$ to $\partial \mathbb{X}_{t_k^m}$ satisfy
	\begin{equation}
	\min_{y_1 \in \partial \mathbb{X}_{s}} \vecnorm{x_1-y_1} = 
	\min_{y_2 \in \partial \mathbb{X}_{s}} \vecnorm{x_2-y_2}, \forall x_1,x_2 \in \partial \mathbb{X}_{t_k^m}.
	\end{equation}
	From the definition of the sets $\mathbb{X}_s$ in \eqref{eq:infTightening} we define this distance $g : \mathbb{R}^+ \times \mathbb{R}^+ \rightarrow \mathbb{R}^+$ between $\mathbb{X}_S$ and $\mathbb{X}_{t_k^m}$ as
	\begin{multline}
	g(s,t_k^m) := 
	\min \big\{(\beta + \vecnorm{M}\hat{v}s)e^{L_x s},\Delta  \big\} \\- \min \big\{(\beta + \vecnorm{M}\hat{v}t_k^m)e^{L_x t_k^m},\Delta \big\},
	\end{multline}
	where,
	\begin{equation}
	\beta :=  \vecnorm{ \sum_{T \in  \mathcal{T}_{k+1}^m}\hat{\xi}}_M	+ \vecnorm{M}2\hat{\theta}
	\end{equation}
	From Assumption~\ref{ass:Nonempty} on the non-emptiness of $\mathbb{X}_s$, $g(\cdot)$ is always well defined.
	We now consider the approximate state trajectory $\tilde{x}(\cdot)$ at times $t_k^m$ and $s\geq t_k^m$. 
	Then,
	\begin{align}
	\norm{d_{\mathbb{X}_{s}}(\tilde{x}(s)) - d_{\mathbb{X}_{t_k^m}}(\tilde{x}(t_k^m))} &\leq \vecnorm{\tilde{x}(s)-\tilde{x}(t_k^m)} + g(s,t_k^m) \\ 
	&= \vecnorm{\int_{t_k^m}^s\dot{\tilde{x}}(t)\mathrm{d}t} + g(s,t_k^m) \\
	& \leq \int_{t_k^m}^{s} \vecnorm{\dot{\tilde{x}}(t)} \mathrm{d}t + g(s,t_k^m).
	\end{align}
	Application of the triangle inequality to the mesh refinement conditions \eqref{eq:RelLocErrBound} yields 
	\begin{equation}
	\begin{split}
	\int_{T_k} \vecnorm{\dot{\tilde{x}}(t)}\mathrm{d}t - 
	\int_{T_k} \vecnorm{f(t,\tilde{x}(t),u(t))}\mathrm{d}t \leq \\ 
	\int_{T_k} \vecnorm{\dot{\tilde{x}}(t)-f(t,\tilde{x}(t),u(t))}\mathrm{d}t \leq \vecnorm{\hat{\xi}},
	\end{split}
	\end{equation}
	from which
	\begin{equation}
	\begin{split}
	\int_{t_k^m}^s \vecnorm{\dot{\tilde{x}}(t)}\mathrm{d}t &\leq \int_{T_k} \vecnorm{\dot{\tilde{x}}(t)}\mathrm{d}t \leq \vecnorm{\hat{\xi}} + \int_{T_k} \vecnorm{f(t,\tilde{x}(t),u(t))}\mathrm{d}t \\ &\leq \vecnorm{\hat{\xi}} + L_f\norm{T_k}.
	\end{split}
	\end{equation}
	Therefore we may bound the distance as
	\begin{equation}
	d_{\mathbb{X}_{s}}(\tilde{x}(s)) - d_{\mathbb{X}_{t_k^m}}(\tilde{x}(t_k^m)) \leq  \vecnorm{\hat{\xi}}+ L_f\norm{T_k} + g(s,t_k^m).
	\end{equation}
	We now need to prove that there is no $s < t_{k+1}^m$ for which $d_{\mathbb{X}_{s}}(\tilde{x}(s)) = 0$. Following the argument in \cite{fontes2018guaranteed}, assume $\exists s_1 < t_{k+1}^m$ such that 
	\begin{equation}
	d_{\mathbb{X}_{s_1}}(\tilde{x}(s_1)) = 0 \text{ and } d_{\mathbb{X}_{s}}(\tilde{x}(s)) > 0, \forall s \in (s_1,t_{k+1}^m].
	\end{equation}
	Therefore, 
	\begin{align}
	d_{\mathbb{X}_{s_1}}(\tilde{x}(s_1)) &\leq d_{\mathbb{X}_{t_k^m}}(\tilde{x}(t_k^m)) +  \vecnorm{\hat{\xi}} + L_f(s_1-t_k^m) + g(s_1,t_k^m) \\ 
	&\leq -\alpha_k +  \vecnorm{\hat{\xi}} + L_f(s_1-t_k^m) + g(s_1,t_k^m)) \\
	& < -\alpha_k + \vecnorm{\hat{\xi}} + L_f\norm{T_k} + g(t_k^m,t_{k+1}^m) \leq 0.
	\end{align}
	This contradicts our assumption. Therefore we may conclude that
	\begin{equation}
	d_{\mathbb{X}_s}(\tilde{x}(s)) \leq 0, \forall s \in [t_k^m, t_{k+1}^m], \forall k \in \mathcal{K}_h.
	\end{equation}
	Combining this with the results of Lemma~\ref{lemma:ConstraintTightening} guarantees
	\begin{equation}
	d_\mathbb{X}(x(t)) \leq 0, \forall t \in [t_i^u,t_{i+1}^u].
	\end{equation}
\end{proof}

Theorem~\ref{th:tight} shows that the constraints enforced at each mesh point $t_k^m$ should be tightened by an amount that is affine in the length of the succeeding mesh segment $\norm{T_k}$ and dependent on the metric of allowable ODE errror $\hat{\xi}$.
This guarantees that the approximate trajectory is admissible for all time. The constraints must be further tightened by an amount $g(\cdot)$ that relates the difference between approximate and actual trajectories.

We could use Theorem~\ref{th:tight} to propose an OCP solution scheme where the constraints are successively loosened at each iteration of the refinement. As the mesh becomes increasingly refined, the constraints will become more and more relaxed. In this way constraint satisfaction of the true state $x(t)$ would be guaranteed at every iteration of the refinement procedure. However, we find in practice that this procedure is particularly conservative, especially for very coarse meshes. This may affect the feasibility of the NLP solved in the initial iterations of the refinement procedure. We could instead start with a finer mesh, but this conversely has adverse effects on solution time. 

Instead we consider a two-part procedure detailed in Algorithm~\ref{alg:MR}. From the condition in Theorem~\ref{th:tight},
\begin{equation}
     \norm{d_{\mathbb{X}_{s}}(\tilde{x}(s)) - d_{\mathbb{X}_{t_k^m}}(\tilde{x}(t_k^m))} \leq \vecnorm{\tilde{x}(s)-\tilde{x}(t_k^m)} + g(s,t_k^m),
\end{equation}
we can guarantee that $d_{\mathbb{X}_{t_{k+1}^m}}(\tilde{x}(t_{k+1}^m)) \leq 0$ if
\begin{equation} \label{eq:tight_iter}
    d_{\mathbb{X}_{t_{k}^m}}(\tilde{x}(t_{k}^m)) \leq \underbrace{\max_{s \in T_k} \vecnorm{\tilde{x}(s) - \tilde{x}(t_k^m)} + g(t_{k+1}^m,t_k^m)}_{\Upsilon(\tilde{x},T_k)}.
\end{equation}
We can similarly guarantee that $d_{\mathbb{X}_{t_{k}^m}}(\tilde{x}(t_{k}^m)) \leq 0$
if 
\begin{equation} \label{eq:tight_iter2}
d_{\mathbb{X}_{t_{k+1}^m}}(\tilde{x}(t_{k+1}^m)) \leq \Upsilon(\tilde{x},T_k).
\end{equation}
In practice, the difference is that \eqref{eq:tight_iter} tightens constraints forwards in time, while \eqref{eq:tight_iter2} passes through the constraints backwards in time.

The left and right hand sides of \eqref{eq:tight_iter} and \eqref{eq:tight_iter2} can be calculated from the approximate prediction $\tilde{x}(t)$. In Algorithm~\ref{alg:MR} we propose first iterating through a typical mesh-refinement procedure (lines \ref{line:MR1}--\ref{line:MR2} in Algorithm~\ref{alg:MR}) to generate a prediction $\tilde{x}$ satisfying ODE error certificates in each segment.
This sufficiently accurate prediction is then used as a proxy to evaluate \eqref{eq:tight_iter} on the refined mesh (lines \ref{line:t1}--\ref{line:t2} in Algorithm~\ref{alg:MR}).
If \eqref{eq:tight_iter} is not satisfied at any time $t \in \mathcal{T}_k$ then the constraint set is tightened by
\begin{equation}
    \mathcal{B}_{\Upsilon_k} := \left\{x \in \mathbb{R}^n : \vecnorm{x} \leq \Upsilon(\tilde{x},T_k) \right\},    
\end{equation}
as in line~\ref{line:BallTight}, and the resulting NLP is solved. We may iterate our algorithm based on either the forwards inequality \eqref{eq:tight_iter} or backwards inequality \eqref{eq:tight_iter2}. 
This procedure is repeated until \eqref{eq:tight_iter} is satisfied $\forall t \in T_h$.
Intuitively, we are robustifying the problem formulation to inaccuracies in our solution method.
This procedure is exemplified in Section~\ref{sec:NonlinearExample}.

At this point we would like to briefly discuss input constraints.These are often simple box constraints of the form 
\begin{equation} \label{eq:input_box_constraints}
u(t) \in \mathcal{U} := \{u \in \mathbb{R}^m : \underline{u} \leq u \leq \overline{u}\},
\end{equation}
which may, for example, arise due to physical actuator limitations.
On the one hand, in our analysis we have assumed that the approximate solution $u(t)$ is implemented exactly at the plant. So for the input we do not need to perform any equivalent of the constraint tightening detailed in Lemma~\ref{lemma:ConstraintTightening}, which bounds the distance between the approximate and actual state trajectories.	

Instead we could seek to just guarantee that $u(t)$ satisfies \eqref{eq:input_box_constraints} between the mesh points --- which is similar in aim to the second part of the constraint tightening procedure introduced in Theorem~\ref{th:tight}. However, additional complications arise because Problem \eqref{eq:Bolza_cont} places no continuity constraints on the input. How these potential discontinuities are handled by the transcription and solution method determines how the constraints must be tightened in order to guarantee continuous time satisfaction of \eqref{eq:input_box_constraints} $\forall t \in T_{op}$.

For example, it may be the case that our optimal control software only allows for input discontinuities at (a particular subset of) mesh points. This is one approach taken that can be taken in ICOCS2\cite{nie2018iclocs2} with the `multi-phase' problem formulation. If we therefore know that $u(t)$ is going to be piecewise polynomial between mesh points then we may use a similar approach as in Theorem~\ref{th:CollTrig}, where a Lipschitz constant of the input representation polynomial --- typically an order below the polynomial used for state representation --- can be used to bound the maximum change in input between the mesh points. We have found this approach to be overly conservative in the tightening unless an extremely dense mesh is used.
Additionally if a low order $p$-method (trapezoidal or Euler) is used then input constraint tightening is unnecessary. Hence we do not tighten any input constraints in the example in Section~\ref{sec:NonlinearExample}

\begin{algorithm}[tb]
    \caption{Adaptive Constraint Mesh Refinement Scheme}
    \label{alg:MR}
    \begin{algorithmic}[1] 
        \Procedure{Mesh Refinement}{\ref{eq:Bolza_cont}, $\hat{\epsilon}$, $\mathcal{T}_h^0$, $J$}
            \State $\epsilon_{k,i }\gets \infty, \ \forall k \in \mathcal{K}, i \in \{1,\hdots,n\}$ 
            \State $j \gets 0$
            \State $\mathbb{X}_s \gets \mathbb{X}\ominus\mathcal{B}_s$
            \While{$\left( \exists \epsilon_{k,i} > \hat{\epsilon}_i\right)$} \label{line:MR1}
                \State $($\ref{eq:Bolza_cont}$_\text{NLP}^j)\gets$ Transcribe$(($\ref{eq:Bolza_cont}$_{\mathbb{X}_s}),\mathcal{T}_h^j)$ 
                \State $\tilde{x}(t),u(t) \gets \text{argmin }($\ref{eq:Bolza_cont}$_\text{NLP}^j)$
                \State Calculate $\epsilon_{k,i}, \ \forall k \in \mathcal{K}, i \in \{1,\hdots,n\}$
                \State $\mathcal{T}_h^{j+1} \gets \text{MeshRefinement} (\mathcal{T}_h^j,\epsilon)$
                \State $j \gets j+1$
            \EndWhile \label{line:MR2}
            \While{$\left(\exists t \in T_h : d_{\mathbb{X}_{t_{k}^m}}(\tilde{x}(t_{k}^m))\geq \Upsilon(\tilde{x}(\cdot),T_k)\right)$ } \label{line:t1}
                \State $\mathbb{X}_{t_k^m} \gets \mathbb{X}\ominus(\mathcal{B}_{t_k^m}\oplus\mathcal{B}_{\Upsilon_k})$ \label{line:BallTight}
                \State $($\ref{eq:Bolza_cont}$_\text{NLP}^j)\gets$ Transcribe$(($\ref{eq:Bolza_cont}$_{\mathbb{X}_{t_k^m}}),\mathcal{T}_h^j)$ 
                \State $\tilde{x}(t),u(t) \gets \text{argmin }($\ref{eq:Bolza_cont}$_\text{NLP}^j)$
                \State $j \gets j+1$
            \EndWhile \label{line:t2}
            \State \textbf{return} $\tilde{x}(t),u(t),\epsilon$
        \EndProcedure
    \end{algorithmic}
\end{algorithm}

\section{Example: Linear autonomous system} \label{sec:LinearExample}
Consider the LTI autonomous system described by
\begin{equation} \label{eq:LinSys}
\dot{x}(t) = Ax(t) + E\hat{w},
\end{equation}
\begin{equation} \label{eq:LinSysPar}
x(t_0) = \begin{bmatrix} 10 \\ 10 \end{bmatrix}, \
A = \begin{bmatrix}
0.05 & 0.5 \\
0 & -0.5
\end{bmatrix}, \
E = \begin{bmatrix}
1 \\
1
\end{bmatrix}, \ \hat{w}=0.01,
\end{equation}
simulated over $6$,s.
We discretize $x(t)$ on a uniform coarse mesh with $h=1.5$\,s using either Forward Euler (FE),
or a third-order Hermite-Simpson (HS) scheme.
The FE discretization corresponds to a piece-wise linear trajectory parametization,
while the HS parametization is piece-wise cubic, and detailed in \cite{kelly2017introduction}.
Fig.~\ref{fig:LinState} shows both approximations alongside the `true' and `nominal' trajectories $x(t),x^*(t)$ generated from Matlab's \texttt{linsim} function.
\begin{figure}[tb]
	\centering
	\includegraphics[width=\columnwidth]{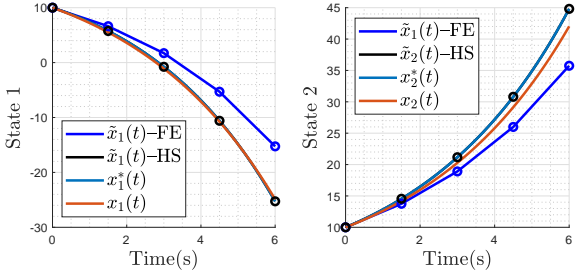}
	\caption[TODO]{A comparison of the true $x(t)$, the nominal $x^*(t)$ and the approximate $\tilde{x}(t) \in \mathcal{X}_{1}$ resulting from FE and HS collocation on a mesh with $h=1.5$, $\sigma=1$.}
	\label{fig:LinState}
\end{figure}
\begin{figure}[t!]
	\centering
	\vspace{0.5em}
	\subfloat[Approximation errors and quadratures. ]{
		\label{subfig:error1}
		\includegraphics[width=\columnwidth]{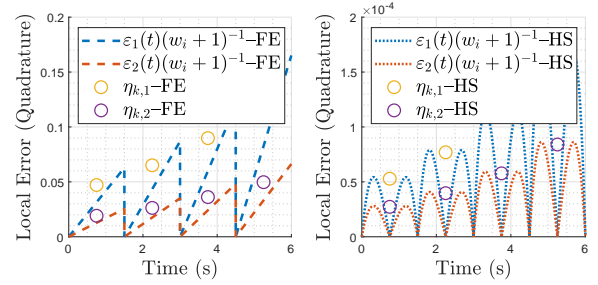}}

	\subfloat[Prediction errors of collocation schemes. ]{
		\label{subfig:error2}
		\includegraphics[width=\columnwidth]{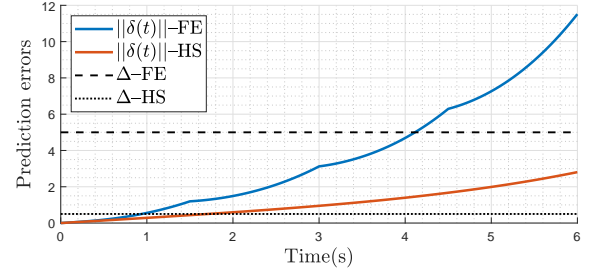} }
	\caption{Plots of the approximation and prediction errors associated to the trajectories shown in Fig.~\ref{fig:LinState}.}
	\label{fig:LinErr}
\end{figure}

As seen in Fig.~\ref{subfig:error1}, the approximation error $\epsilon_i(t)$ and quadratures $\eta_{k,i}$ for the HS approximation are multiple orders of magnitude lower than for FE on the same mesh. For each scheme we a posteriori choose $\hat{\eta}= \max_k \eta_k$, which equals $3.8420$ in the case of FE and $0.0081$ for HS.
Since system \eqref{eq:LinSys} is unstable, the states in Fig.~\ref{fig:LinState} diverge and corresponding approximation errors grow.

Corresponding to the different approximation errors, we chose an ETC threshold of $\hat{\epsilon}=5$ for the FE and $\hat{\epsilon}=0.5$ for the HS approximation. In both cases $\hat{\epsilon}>\hat{\eta}$ and the ETC condition \eqref{eq:EventTriggeringCondition} is met after $4.12$\,s and $1.7$\,s, respectively. Fig.~\ref{subfig:error2} shows that the FE approximation error results in significant growth in $\epsilon(t)$. The minimum IUT $\underline{\tau}_u$ and times $\tau^\text{QET}, \tau^\text{CT}$ are detailed in Table~\ref{tab:1} for both approximations, and the associated functions plotted in Fig.~\ref{fig:LinTrig}.

\begin{table}[tb]
	\centering
	\caption{Comparison of estimated IUTs.}
	\begin{tabular}{ l|c|c|c|c|c }
		& $\hat{\eta}$ & $\underline{\tau}^u$ & $\tau^\text{QET}(\eta)$ & $\tau^\text{QET}(\varepsilon)$ & $\tau^\text{CT}$ \\
		\hline\hline
		FE $(\hat\epsilon=5)$ & $0.1$ & $0.25$ & $1.49$ & $1.74$ & $1.39$ \\
		HS $(\hat\epsilon=0.5)$ & $ 0.045$ & $1.22$  & $ 1.27$ & $1.28$ & $0.4$ \\
	\end{tabular} \label{tab:1}
\end{table}

\begin{figure}[tb]
	\centering
	\includegraphics[width=\columnwidth]{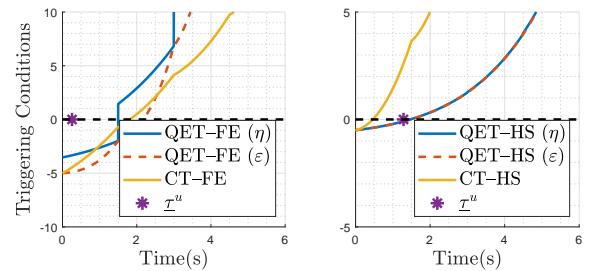}
	\caption[TODO]{A comparison of the online triggering estimations proposed in Section~\ref{sec:OnlineTrig} for FE and HS Schemes.}
	\label{fig:LinTrig}
\end{figure}

The line segments of $\text{QET-FE}(\eta)$ and $\text{QET-HS}(\eta)$ in Fig.~\ref{fig:LinTrig} have slope $L_x \hat{\epsilon} + \hat{w}$, and reflect state uncertainty due to $w(t)$. The discontinuities arise from the approximation errors. For the FE scheme, $\underline{\tau}^u$ is significantly more conservative than $\tau^\text{QET}(\eta)$, whereas $\underline{\tau}^u \approx \tau^\text{QET}(\eta)$ for the HS. This is because $\underline{\tau}^u$ is estimated from the worst case absolute local error. From Fig.~\ref{fig:LinErr} the absolute difference between the errors in the first and last segment is orders of magnitude more than the difference between the errors in the first and last segment of the HS scheme.
Using relative local error, as suggested in \cite[Ch. 5]{betts2010practical}, may result in better agreement between $\underline{\tau}$ and $\tau^\text{QET}(\eta)$ for the FE scheme.

\section{Example: Closed-loop Simulation} \label{sec:NonlinearExample}

We now consider a closed-loop example of the two link robotic arm, shown in Figure~\ref{fig:TwoLinkRobotArm}, over a simulation time of~$8$\,s. 
\begin{figure}[tb]
	\centering
	\includegraphics[width=\columnwidth]{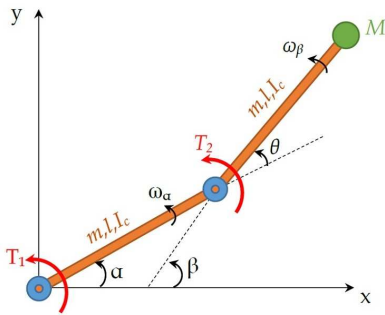}
	\caption[TODO]{The two link robot arm plant considered in closed-loop simulations}
	\label{fig:TwoLinkRobotArm}
\end{figure}
We employ the quadratic objective function
\begin{equation}
\int_{t_0}^{t_f} u_1^2(t) + u_2^2(t) \mathrm{d}t,
\end{equation}
over a receding horizon of length $T_\text{OP}=4$\,s. The input $u(t)\in\mathbb{R}^2$, and the state is  $x:=[\omega_\alpha, \   \omega_\beta, \ \theta, \ \beta]^\intercal$. $\omega_\alpha$ and $\omega_\beta$ are the angular accelerations of each joint in the arm and $\theta, \ \beta$ the corresponding angular velocities. The nonlinear time-invariant dynamics $f(x,u)$ are described by
\begin{align}
    \dot\omega_\alpha(t) &= \Big(\frac{9}{4}\sin(\theta(t))\cos(\theta(t))x_1(t)^2 + 2x_2(t)^2 +
    \frac{4}{3}u_1(t) - \\ & \quad \quad \frac{4}{3}u_2(t) - \frac{3}{2}\cos(\theta(t))u_2(t)\Big)
    \Big(\frac{31}{36} + \frac{9}{4}\sin^2(\theta(t))\Big)^{-1}, \\
    \dot\omega_\alpha(t) &= \Big(\frac{9}{4}\sin(\theta(t))\cos(\theta(t))x_2(t)^2 + \frac{7}{2}x_1(t)^2 - \frac{7}{3}u_2(t) + \\ & \frac{3}{2}\cos(\theta(t))\left(u_1(t)-u_2(t)\right)\Big) \Big(\frac{31}{36} + \frac{9}{4}\sin^2(\theta(t))\Big)^{-1}, \\
    \dot\theta(t) &= \omega_\alpha(t) - \omega_\beta(t), \\
    \dot\beta(t) &= \omega_\beta(t).
    \label{eq:floatingeq}
\end{align}
Each state is additionally effected with additive noise $v(t) \in \mathbb{R}^4, \vecnorm{v(t)}\leq 5\times 10^3$.
Using the method described in \cite{khalil2002nonlinear}, we estimate the Lipschitz constant of the system as $L_x=4.5309$.
We employ a trapezoidal collocation, detailed in \cite{kelly2017introduction}, where linear input and quadratic state paremterisation is assumed. 

We base the mesh refinement steps in Algorithm~\ref{alg:MR} on the algorithm in \cite[Ch.5]{betts2010practical}, with refinement based on metric \eqref{eq:AbsoluteQuadrature}. This initial refinement phase also terminates when \eqref{eq:AbsoluteQuadrature} is less than $\hat{\epsilon}_i=10^{-3}$ in each mesh segment, after which we perform constraint tightening.
We constrain the angular acceleration of the first joint to be $\omega_\alpha\in[-0.3,0.3]$. A viable distance operator may be used with the presence of additional constraints through the method proposed in~\cite{fontes2018guaranteed}.

We consider two different closed-loop update scenarios. In scenario a) the ETC scheme is used to trigger control updates based on the one-norm of the difference between the predicted state $\tilde{x}(t)$ and measured state $x(t)$ using condition \eqref{eq:EventTriggeringCondition} with threshold $\Delta=7\times10^{-3}$. 
In scheme b) we pose a self-triggered control setup, where updates are performed at times in accordance with Theorem~\ref{th:QET}, also with threshold $\Delta=7\times10^{-3}$. 
The closed-loop input and state trajectories are shown in Figures~\ref{fig:NL_Input} and \ref{fig:NL_State}. The subscripts $E$ and $S$ are respectively used to denote the closed-loop trajectory for the event- and self-triggered update schemes. 

\begin{figure}[tb]
	\centering
	\includegraphics[width=\columnwidth]{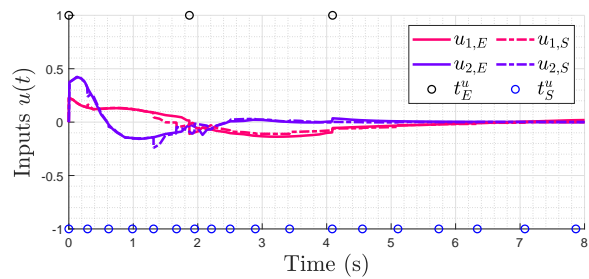}
	\caption[TODO]{Closed-loop input trajectories and control update times for the ETC scheme and STC scheme based on Theorem~\ref{th:QET}.}
	\label{fig:NL_Input}
\end{figure}

\begin{figure}[tb]
	\centering
	\includegraphics[width=\columnwidth]{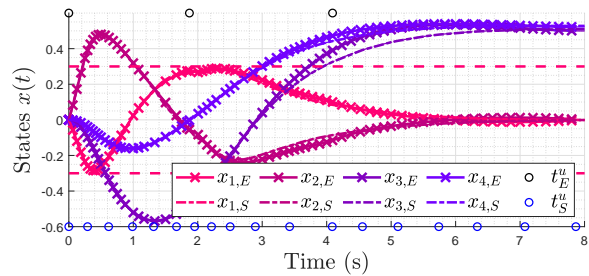}
	\caption[TODO]{Closed-loop state trajectories and control update times for the ETC scheme and STC scheme based on Theorem~\ref{th:QET}. The crosses depict the combined NLP solution of the ETC solves.}
	\label{fig:NL_State}
\end{figure}

The crosses in Figure~\ref{fig:NL_State} are used to show the location of the mesh points of the concatenated closed-loop solution. The update times $\mathcal{T}_u$ for the ETC and STC scenarios are shown by the black and blue open circles. The IUTs of the ETC scheme are notably less significant than the STC IUTs, with updates on average every $2.69$\,s. Conversely, condition \eqref{eq:trig_QET} results in updates on average every $0.48$\,s. 
The closed-loop cost of the ETC scheme is $0.1335$, while the increased triggering of the STC scheme costs $0.1243$.
However, it is important to note that \eqref{eq:trig_QET} does not use any knowledge of the predicted or measured state, but only on the error in the solution of the OCP guaranteed by the mesh refinement algorithm. Even if the  triggering is based just on condition \eqref{eq:trig_QET}, the IUTs increase as the state approaches the terminal condition. Future work should consider a condition that depends on both the predicted state $\tilde{x}$ \emph{and} the quadrature errors $\epsilon_k$.

The proposed triggering does account for proximity to constraints. Using Theorem~\ref{th:tight} we can guarantee constraint satisfaction of the continuous-time plant in both the event- and self-triggered schemes. At each update we apply the two-stage tightening procedure in Algorithm~\ref{alg:MR}. The result of the first solve at time $t=0$ is shown in Figure~\ref{fig:NL_ConstTight} for the constrained state $\omega_\alpha(t)$. The dashed red line shows the first stage of the procedure, where constraint tightening is based on $\mathcal{B}_s$. The black line shows the predicted state trajectory using only this first stage of tightening. As shown, the quadratic sate interpolation results in the predicted state violating the tightened constraint. In this case we cannot guarantee constraint satisfaction of the real plant at these times. 
The second stage of tightening, at discrete times $\mathcal{T}_m$, is shown by the blue dots. For clarity of representation, we mimic the tightening of the upper and lower bounds. Nominally, for a sufficiently dense mesh the trajectory of the tightened problem and non-tightened problem are similar. 

\begin{figure}[tb]
	\centering
	\includegraphics[width=\columnwidth]{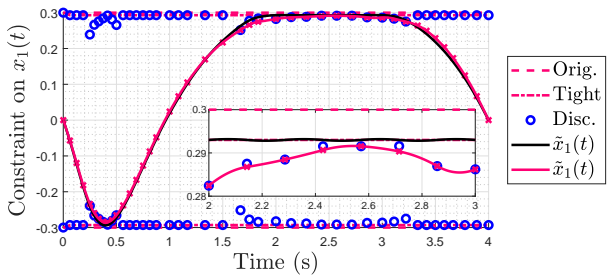}
	\caption[TODO]{An example of the constraint tightening mesh-refinement procedure proposed in Section~\ref{sec:constraintTightening} and Algorithm~\ref{alg:MR}. The black line shows the state refined state trajectory after the constraint set is tightened from $\mathbb{X}$ to $\mathbb{X}_s$. The pink link shows the final trajectory after the constraints are additionally iteratively tightened. The blue dots show the final constraint bounds at the mesh points.}
	\label{fig:NL_ConstTight}
\end{figure}

Finally, in Figure~\ref{fig:CompComparison} we consider how different combinations of thresholds $(\Delta,\hat\epsilon)$ result in both a different  load and distribution in time of computational load for the ETC scheme. As the threshold $\Delta$ increases, the number of control updates and the time spent in IPOPT decreases. We also see that as the quality of the prediction improves, the number of updates decreases (up to a point), but the total time in IPOPT tends to increase.

\begin{figure}[tb]
	\centering
	\includegraphics[width=\columnwidth]{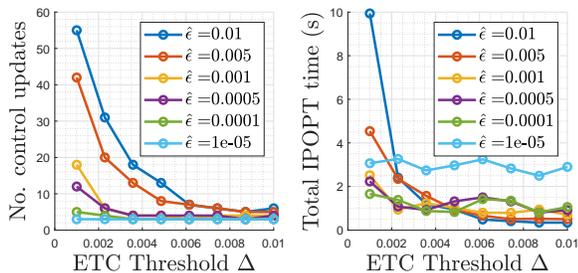}
	\caption[TODO]{Comparison of computation time and number of control updates for various ETC and refinement conditions $(\hat{\epsilon},\Delta)$}
	\label{fig:CompComparison}
\end{figure}

\section{Conclusions}
Event- and self-triggering schemes attempt to reduce performance loss due to the growth in uncertainty as the system evolves in open-loop. Typically, uncertainty arising in the physical world may  at best be unmodelled and at worst unstructured. However, uncertainty arising due to the solution of \eqref{eq:Bolza_cont} is highly structured, therefore it should not be treated equivalently to, say, $v(t)$.

This is the first work that has considered using the accuracy of the solution of an optimal control problem as a metric for triggering times. This gives us the potential to use cheaper, less accurate, solutions without potentially sacrificing the guarantees of a triggered model predictive control scheme.
We have considered the use of approximate system predictions arising through numerical direct collocation solutions of optimal control problems in event-triggered control schemes. A commonly used metric for refining the state approximations has allowed us to guarantee a minimum inter-update time $\underline{\tau}^u$ of the event-triggered control. If additionally the direct collocation software returns the approximation error metric as part of the solution, we may achieve a better estimate of the inter-update interval. Although the derived bounds may be conservative, they give us the ability to trade-off the computation burden of a single control update with the frequency of control updates.

This theme is also of importance in simulation of event- or self-triggered control which require potentially unstable numerical simulations to be run in open-loop for extended time steps, without the stabilizing effect of periodic state feedback.
Further work may consider the use of path constraint violations as a metric for mesh refinement and/or triggering.


\bibliography{MR_ETC_bibliography}
\bibliographystyle{ieeetr}

\begin{IEEEbiography}[{\includegraphics[width=1in,height=1.25in,clip,keepaspectratio]{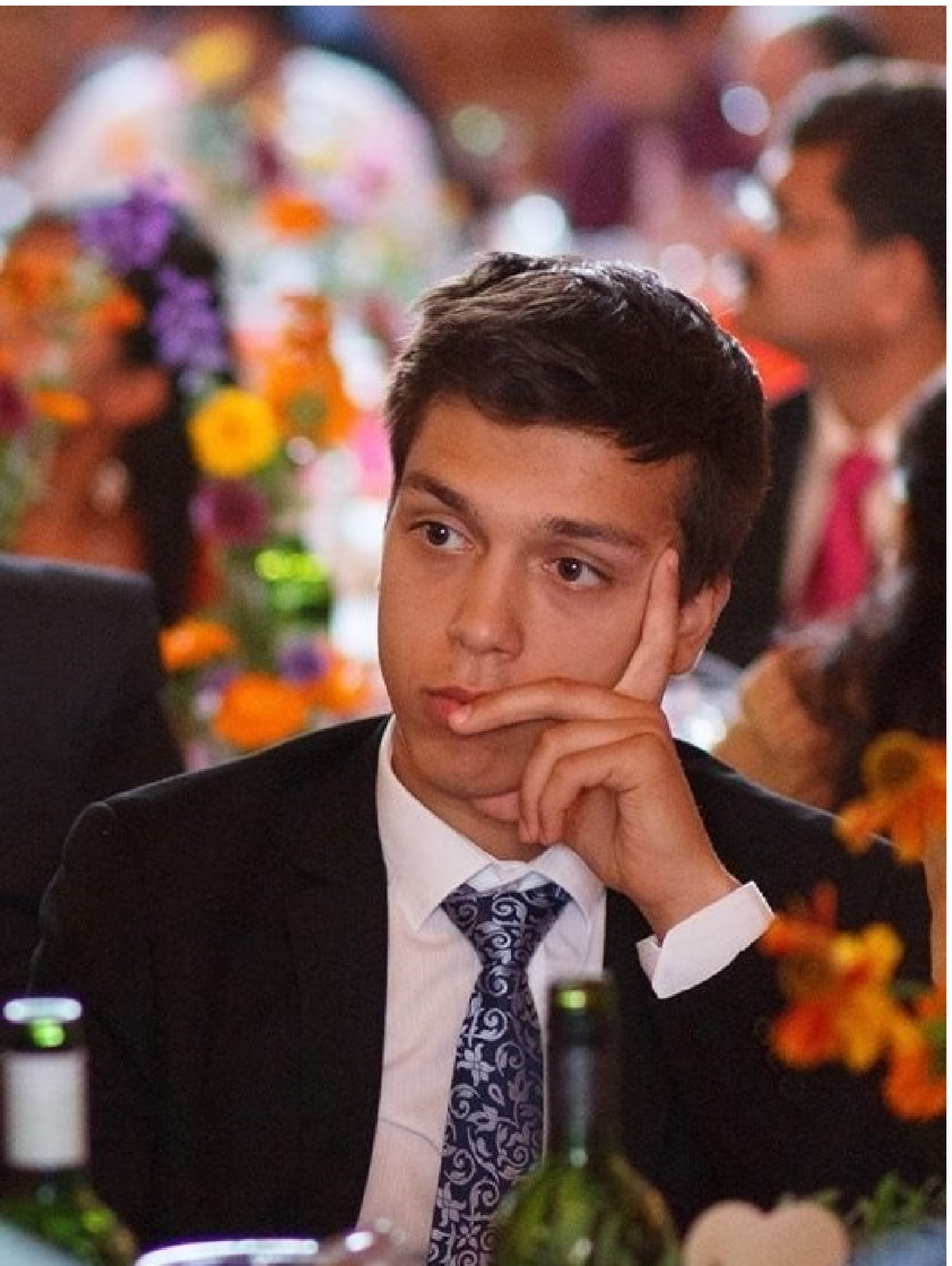}}]{Omar Faqir}
received an MEng from Imperial College London (ICL), London, UK, in 2016. 
He is currently pursuing a PhD in the Electrical and Electronic Engineering Department of ICL, with the Control and Power Research group and the Information Processing and Communications Lab.
His current research interests include optimisation based control theory, wireless communications and autonomous vehicles.
\end{IEEEbiography}

\begin{IEEEbiography}[{\includegraphics[width=1in,height=1.25in,clip,keepaspectratio]{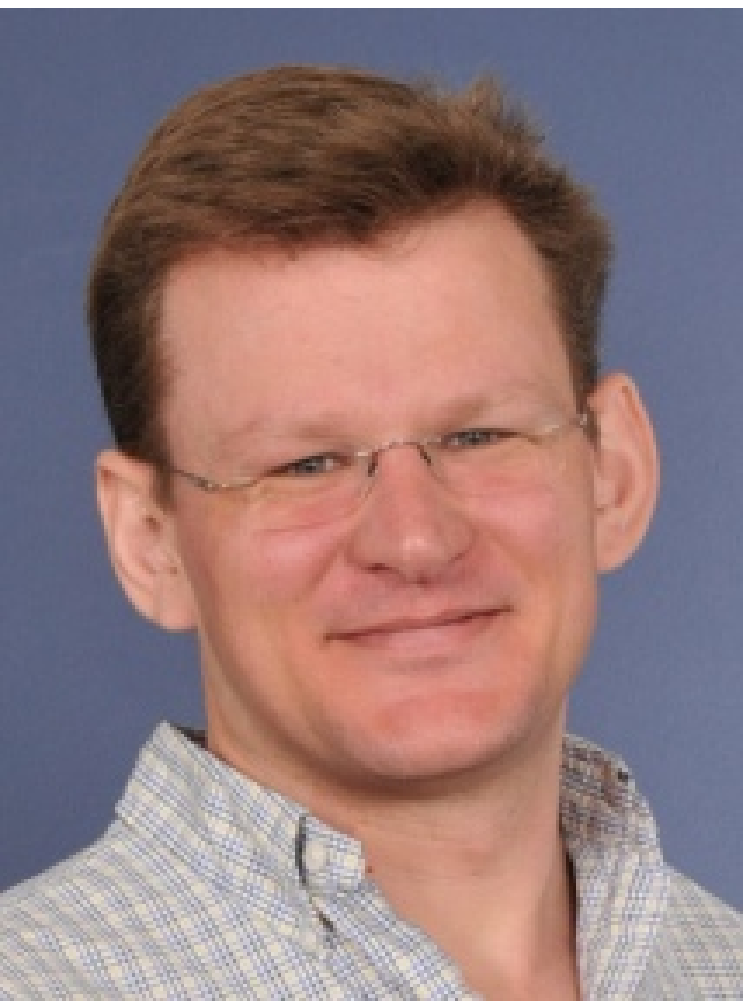}}]{Eric Kerrigan}
(S'94-M'02-SM’16) received a BSc(Eng) from the University of Cape Town and a PhD from the University of Cambridge. His research is in the design of efficient numerical methods and computing architectures for solving  optimal control problems in real-time, with applications in the design of aerospace, renewable energy and information systems. He is  the chair of the IFAC Technical Committee on Optimal Control, a Senior Editor of IEEE Transactions on Control Systems Technology and  an Associate Editor of IEEE Transactions on Automatic Control and the European Journal of Control.
\end{IEEEbiography}

\end{document}